\documentclass[sigconf]{acmart}

\usepackage{algorithm}
\usepackage{algorithmic}

\usepackage{threeparttable}
\usepackage{multirow}
\usepackage{adjustbox}
\usepackage{wasysym} 
\usepackage{subcaption}
\usepackage{soul}
\usepackage{tabularx}
\usepackage{colortbl}
\usepackage{xcolor}
\usepackage{enumitem}
\newcommand{\cmark}{\CIRCLE}
\newcommand{\xmark}{\Circle}

\newtheorem{assumption}{Assumption}

\DeclareMathOperator*{\esssup}{ess\,sup}
\DeclareMathOperator*{\argmax}{arg\,max}

\newcommand{\R}{\mathbb{R}}
\newcommand{\E}{\mathbb{E}}
\newcommand{\Prb}{\mathbb{P}}
\newcommand{\ind}{\mathbf{1}}
\newcommand{\cX}{\mathcal{X}}
\newcommand{\cD}{\mathcal{D}}
\newcommand{\cQ}{\mathcal{Q}}

\newcommand{\asr}{\operatorname{ASR}}
\newcommand{\tgg}{\operatorname{TGG}}
\newcommand{\caacc}{\operatorname{CA}}

\AtBeginDocument{%
  }

\acmSubmissionID{919}

\begin{document}

\title{Lilith: Backdoor Generalization under Training–Inference Trigger Shift}

\author{Zhou Feng}
\orcid{0009-0006-5301-7019}
\affiliation{%
  \institution{Zhejiang University}
  \department{College of Computer Science and Technology}
  \city{Hangzhou}
  \country{China}
}
\email{zhou.feng@zju.edu.cn}

\author{Jiahao Chen}
\orcid{0000-0002-5894-662X}
\affiliation{%
  \institution{Zhejiang University}
  \department{College of Computer Science and Technology}
  \city{Hangzhou}
  \country{China}
}
\email{xaddwell@zju.edu.cn}

\author{Chunyi Zhou}
\orcid{0000-0003-0081-0946}
\affiliation{%
  \institution{Zhejiang University}
  \department{College of Computer Science and Technology}
  \city{Hangzhou}
  \country{China}
}
\email{zhouchunyi@zju.edu.cn}

\author{Yuan Su}
\orcid{0000-0002-5121-9331}
\affiliation{%
  \institution{Zhejiang University}
  \department{College of Computer Science and Technology}
  \city{Hangzhou}
  \country{China}
}
\email{yuansu@zju.edu.cn}

\author{Tianyu Du}
\orcid{0000-0003-0896-0690}
\affiliation{%
  \institution{Zhejiang University}
  \department{School of Software Technology}
  \city{Hangzhou}
  \country{China}
}
\email{zjradty@zju.edu.cn}

\author{Yuwen Pu}
\orcid{0000-0003-2311-4943}
\affiliation{%
  \institution{Chongqing University}
  \department{School of Big Data \& Software Engineering}
  \city{Chongqing}
  \country{China}
}
\email{yw.pu@cqu.edu.cn}

\author{Jianhai Chen}
\orcid{0000-0003-3524-3443}
\affiliation{%
  \institution{Zhejiang University}
  \department{College of Computer Science and Technology}
  \city{Hangzhou}
  \country{China}
}
\email{chenjh919@zju.edu.cn}

\author{Jinbao Li}
\orcid{0000-0002-2432-8807}
\authornotemark[1]
\affiliation{%
  \institution{Qilu University of Technology (Shandong Academy of Science)}
  \city{Jinan}
  \country{China}
}
\email{lijb@qlu.edu.cn}

\author{Shouling Ji}
\orcid{0000-0003-4268-372X}
\authornote{Corresponding authors.}
\affiliation{%
  \institution{Zhejiang University}
  \department{College of Computer Science and Technology}
  \city{Hangzhou}
  \country{China}
}
\email{sji@zju.edu.cn}

\renewcommand{\shortauthors}{Zhou Feng et al.}
\acmArticleType{Review}
\acmContributions{BT and GKMT designed the study; LT, VB, and AP
  conducted the experiments, BR, HC, CP and JS analyzed the results,
  JPK developed analytical predictions, all authors participated in
  writing the manuscript.}

\ccsdesc[500]{Security and privacy~Formal methods and theory of security}
\ccsdesc[300]{Computing methodologies~Machine learning}
\ccsdesc[300]{Computing methodologies~Artificial intelligence}

\keywords{Backdoor Learning, Adversarial Machine Learning}

\begin{teaserfigure}
    \centering
    \includegraphics[width=0.9\linewidth]{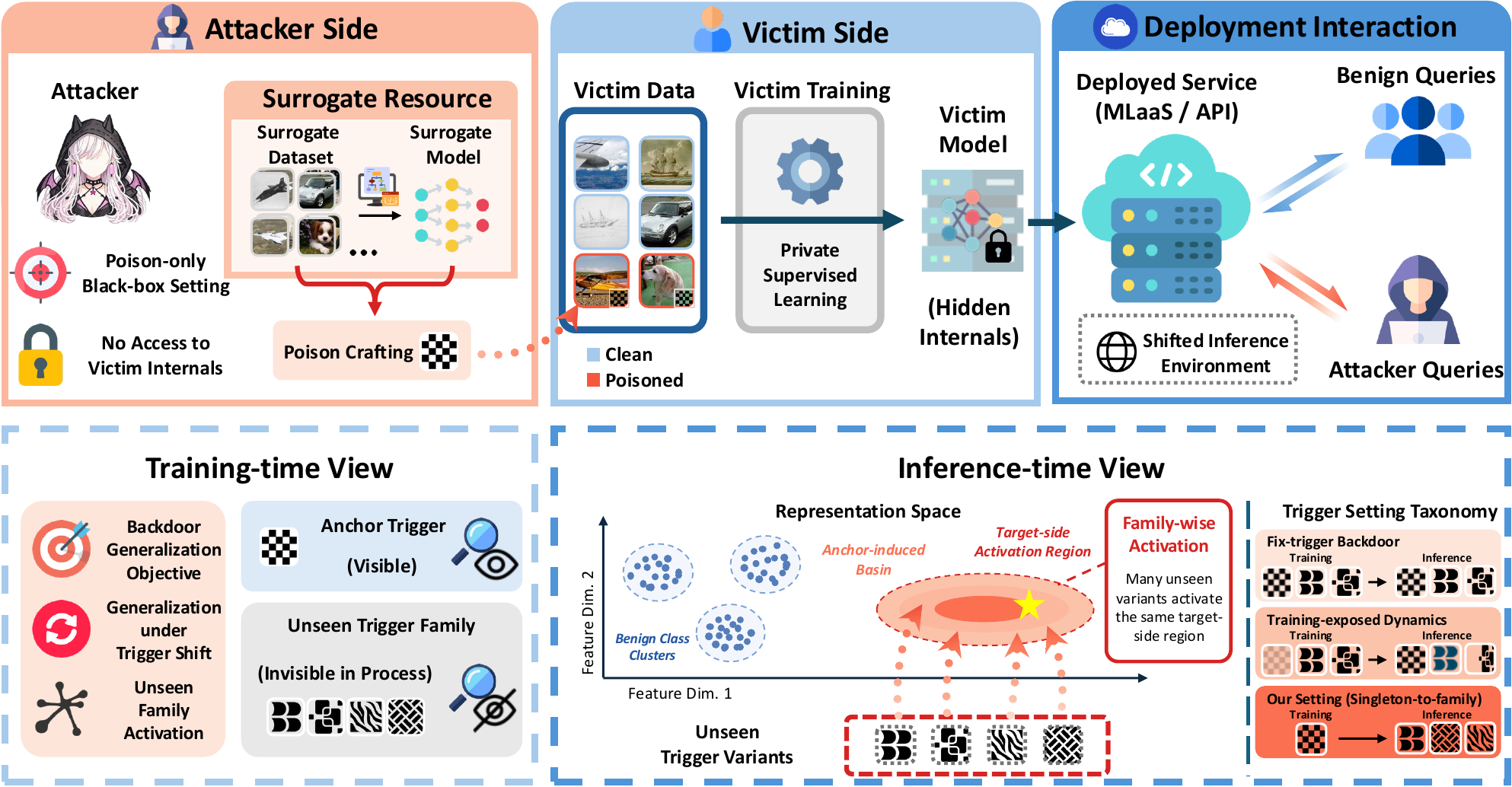}
    \caption{General introduction of \textbf{Lilith}.}
    \label{fig:intro}
\end{teaserfigure}

\begin{abstract}
Machine-learning services increasingly rely on public data, third-party
providers, and outsourced training, creating opportunities for
data-poisoning attacks that implant persistent malicious behavior while
preserving benign utility.
However, existing backdoor studies largely evaluate exact trigger reuse,
training-exposed trigger diversity, or variations along predefined
transformation axes. They therefore leave a critical blind spot:
whether a backdoor learned from one training-time trigger can generalize
to an inference-time trigger family absent from victim training.
We formulate this problem as \emph{backdoor generalization under
training--inference trigger shift} and introduce \textbf{Lilith}, a
black-box anchor-to-family framework. Using only disjoint surrogate
resources, \textbf{Lilith} first induces a compact target-side vulnerability with
a single training anchor, then constructs a bounded inference-only
family that preserves the anchor-induced representation geometry. We
characterize this mechanism through anchor clearance and family reach,
deriving sufficient conditions for family-wise target preservation
under local regularity and bounded surrogate--victim discrepancy.
Experiments across datasets, architectures, poisoning rates, and
defenses show that \textbf{Lilith} achieves high family-wise attack success with
limited utility degradation and a small trigger generalization gap.
Additional analyses show that family activation depends on
representation alignment rather than the proposal mechanism, exposing
a broader threat overlooked by exact-trigger evaluation.
\end{abstract}

\maketitle

\section{Introduction}
Recent advances in machine learning, supported by larger models,
broader training corpora, and increasingly complex learning pipelines,
have substantially improved recognition, detection, generation, and
automated decision-making capabilities
~\cite{Yang2017DepressionRecognition,Nassif2021Anomaly,
Summerville2018PCGML}. These advances have also increased the reliance
of machine-learning systems on public datasets, third-party data
providers, and outsourced training services. By modifying only a small
fraction of the training data, an adversary may implant a backdoor that
preserves normal utility on benign inputs but induces an
attacker-chosen prediction when a designated trigger is presented
~\cite{gu2017badnets}. Such hidden behavior threatens the integrity and
long-term reliability of machine-learning services deployed in
safety-critical and commercial environments
~\cite{Chen2024FaceVulnerabilities,Ge2025LLMBackdoor}.

A successful attack against a practical
machine-learning-as-a-service platform must survive more than victim
training. Production pipelines may inspect, deduplicate, augment,
normalize, or compress data. Deployed services may monitor data quality,
input and prediction drift, model performance, and feature attributions
~\cite{MicrosoftAzureModelMonitoring,GoogleVertexModelMonitoring,
AmazonSageMakerModelMonitor}. Although these mechanisms are not
dedicated backdoor defenses, they increase operational pressure on
attacks that repeatedly expose one recognizable artifact.
Backdoor-specific defenses inspect poisoned samples, runtime
predictions, reconstructed triggers, and compromised representations
~\cite{Tran2018SpectralSignatures,gao2019strip,
wang2019neural,li2021antibackdoor,Min2023FST}.
Consequently, a conventional attack that only memorizes one
poisoning-time pattern may be disrupted by serving-time transformations,
become conspicuous under repeated activation, or lose effectiveness
once the known trigger is identified and filtered, revealing a series of intrinsic limitations that challenge the sustainability of existing backdoor methodologies:


\noindent$\bullet$\quad \textbf{L1. Trigger Reuse.}
Conventional attacks generally evaluate the same trigger mechanism
during poisoning and inference
~\cite{gu2017badnets,zeng2023narcissus}. Their reported success
therefore primarily measures \textbf{exact-trigger activation},
rather than whether the malicious behavior generalizes when the
inference trigger differs from the poisoning artifact.

\noindent$\bullet$\quad \textbf{L2. Training-Exposed Diversity.}
Dynamic and sample-specific attacks generate diverse patterns, but
this diversity is already represented in the victim's poisoned
training data or produced by a mechanism shared across training and
inference~\cite{Nguyen2020InputAware,Li2021ISSBA}.
TITIM directly studies training--inference trigger mismatch and
shows that changes in trigger size or opacity influence attack
generalization and detection~\cite{lin2025titim}. However, its
variation remains centered on predefined intensity dimensions of
the same trigger identity, and its robust configurations may expose
multiple intensities during training.

\noindent$\bullet$\quad \textbf{L3. Missing Preservation Criterion.}
Input-space similarity does not guarantee that two triggers induce
the same latent behavior, particularly when surrogate and victim
environments differ. Constructing an inference-only trigger family
therefore requires a principled criterion that connects each unseen
variant to the vulnerability implanted by the training trigger.

These limitations motivate us to investigate
\textbf{training--inference trigger shift}. Rather than asking whether
a backdoor tolerates small modifications to a known trigger, we study
whether malicious behavior learned from one training-time trigger can
generalize to inference-time triggers never exposed during victim
training. This shifts the focus from
\textbf{trigger-instance robustness} to
\textbf{malicious behavior generalization}.
From a security perspective, such generalization can expand the
post-deployment attack surface from one identifiable trigger instance
to multiple behaviorally aligned activation routes, revealing a
\textbf{higher-freedom threat} that exact-trigger evaluation may
overlook in practical MLaaS systems.
The theoretical novelty is to model trigger variation as a shift in
attack support, rather than exact trigger reuse, training-exposed
diversity, or variation along predefined transformation axes. The
central question is whether unseen trigger diversity can inherit a
vulnerability induced by a single training anchor.

We formulate this problem through representation geometry. A training
anchor should establish a compact target-oriented region with sufficient
decision-boundary clearance, while valid inference variants should
remain within its effective representation reach. Backdoor
generalization is therefore governed by
\textbf{anchor clearance} and \textbf{family reach}. Under a
poison-only black-box setting, we ask whether an attacker can use
disjoint surrogate resources to poison the victim with one trigger,
while constructing an inference-only family that preserves
family-wise activation, benign utility, and perceptual stealth after an
unknown victim training process. This problem introduces two coupled
challenges.

\noindent\textbf{Challenge 1:}
\textit{How can one training-time trigger induce a stable target-oriented
vulnerability when the victim architecture, data, and training process
are unknown?}

\noindent\textbf{Solution 1:}
We introduce \textbf{Anchor Vulnerability Induction}. The module
encourages anchor-triggered representations to form a compact region on
the target side of the decision boundary, rather than optimizing
exact-trigger success alone. Perceptual and frequency constraints
control input-space visibility, while representation concentration and
target separation enlarge vulnerability clearance.

\noindent\textbf{Challenge 2:}
\textit{How can diverse inference-time triggers remain connected to this
vulnerability without participating in victim training?}

\noindent\textbf{Solution 2:}
We introduce \textbf{Representation-Aligned Family Construction}.
Candidate variants are proposed within a bounded subspace and retained
according to their geometric relation to the anchor-induced region.
The proposal source is not essential. Successful variants must preserve
the anchor geometry while maintaining input-space diversity. Our
analysis gives sufficient conditions for target-decision preservation
under local regularity and bounded surrogate--victim discrepancy.

These solutions form \textbf{Lilith}, an
\textbf{anchor-to-family backdoor framework}. Our contributions are
summarized as follows:

\noindent$\bullet$\quad
We introduce \textbf{Lilith}, a black-box anchor-to-family
framework that combines Anchor Vulnerability Induction with
Representation-Aligned Family Construction. Using only disjoint
surrogate resources, it creates a compact target-side anchor basin and
a bounded trigger family that preserves its geometry.

\noindent$\bullet$\quad
We formulate backdoor generalization under
training--inference trigger shift as a support-shift problem. By
characterizing anchor clearance and family reach, we derive sufficient
conditions for unseen family members to preserve the target decision
under bounded surrogate--victim discrepancy.
 
\noindent$\bullet$\quad
We evaluate \textbf{Lilith} across datasets, architectures, poisoning
rates, proposal mechanisms, deployment perturbations, and defenses.
\textbf{Lilith} maintains high family-wise activation with a small
generalization gap, while the results further confirm the importance
of representation alignment and representation-centric mitigation.

\section{Preliminary \& Related Works}

\noindent\textbf{Backdoor Attack.}
Backdoor attacks implant malicious behavior by introducing a small
number of triggered samples into training data
~\cite{Li2024BackdoorLearningSurvey,gu2017badnets}. Early attacks rely
on fixed and reusable patterns
~\cite{gu2017badnets,chen2017targeted}, while later work improves
stealth and practicality through clean-label poisoning
~\cite{turner2019label,zeng2023narcissus,souri2022sleeper}, natural
reflections~\cite{liu2020reflection}, spatial warping
~\cite{nguyen2021wanet}, and latent feature manipulation
~\cite{Yao2019LatentBackdoor,qi2023revisiting}. Frequency-domain
methods such as FTrojan~\cite{Wang2022FrequencyBackdoor} and
WaveAttack~\cite{Xia2024WaveAttack} further construct low-visibility
triggers through structured spectral components. These methods improve
imperceptibility, transferability, or attacker capability, but generally
evaluate activation using the injected trigger or its original
generation rule.

A separate line of work increases trigger diversity. IAD
~\cite{Nguyen2020InputAware} and ISSBA~\cite{Li2021ISSBA} generate
input-aware or sample-specific triggers, although diverse trigger
instances are already exposed during victim training. Trigger-mismatch
research further studies whether a known trigger remains effective
under changes such as size or opacity~\cite{lin2025titim}. These studies
show that backdoor activation need not depend on exact visual
reproduction of a trigger. However, the conditions under which malicious
behavior generalizes across a broader shift in trigger support remain
insufficiently understood. In particular, it is unclear which
representation-level properties enable unseen triggers to preserve the
same malicious decision, and how such generalization should be modeled
and evaluated when inference-time diversity is not fully represented
during training.

\noindent\textbf{Backdoor Defense.}
Production MLaaS platforms increasingly provide monitoring for data quality, input and prediction drift, model performance, and feature attributions~\cite{MicrosoftAzureModelMonitoring,GoogleVertexModelMonitoring,AmazonSageMakerModelMonitor}, improving operational observability. Backdoor-specific detectors instead analyze poisoned-feature statistics~\cite{Tran2018SpectralSignatures,ma2023beatrix}, prediction consistency~\cite{gao2019strip,guo2023scaleup,pal2024backdoor}, or parameter and activation behavior~\cite{hou2024ibdpsc,liu2023detecting}. Inference-time purification methods, including ZIP~\cite{Shi2023ZIP} and SampDetox~\cite{Yang2024SampDetox}, remove trigger information without modifying the deployed model. Model-level mitigation includes trigger inversion~\cite{wang2019neural}, attention distillation~\cite{li2021neural}, adversarial unlearning~\cite{zeng2022adversarial}, fine-tuning-based repair~\cite{zhu2023enhancing}, and representation-oriented methods such as FST~\cite{Min2023FST}, which directly reshape compromised feature geometry. Proactive defenses such as ABL~\cite{li2021antibackdoor} and later robust-training methods~\cite{qi2023proactive,pu2024mellivora} instead intervene during victim training. These representation-level defenses are particularly relevant to trigger families sharing latent behavior rather than a surface pattern.

\section{Problem Definition}
\label{sec:problem}

\begin{figure*}[t]
    \centering
    \includegraphics[width=1.0\linewidth]{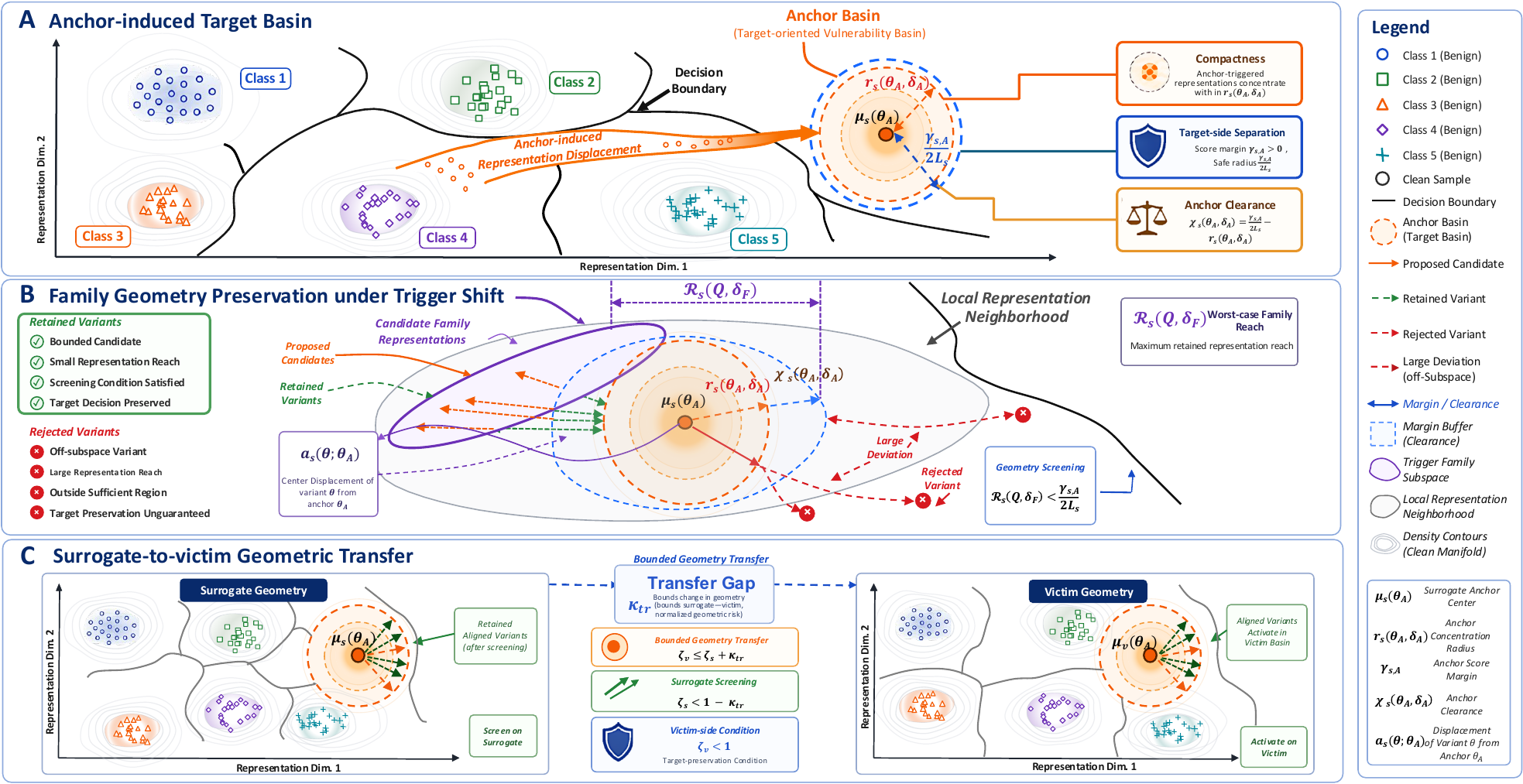}
    \caption{Geometric view of \textbf{Lilith}. The anchor forms a compact target-side basin, and retained variants remain within its margin after bounded surrogate--victim transfer.}
    \Description{A three-panel representation-space diagram showing anchor-basin formation, aligned and misaligned trigger variants, and surrogate-to-victim geometric transfer.}
    \label{fig:lilith-geometry}
\end{figure*}

\noindent\textbf{Motivation.}
Conventional backdoor evaluation typically couples victim poisoning
with exact or training-exposed trigger reuse at inference. Under this
protocol, a high attack success rate may reflect memorization of the
poisoning artifact, but does not reveal whether the implanted malicious
behavior persists when the inference trigger support changes. To
separate artifact-specific activation from behavioral generalization,
we study a support-shift setting in which one training-time trigger
induces a vulnerability that may be activated by previously unseen
inference-time variants. This formulation will enable us to characterize
the problem through the geometry established by the training anchor and
the representation deviation introduced by the inference family.

Before formalizing this trigger shift, we first specify the
\textbf{poison-only black-box setting} that constrains the attacker's
knowledge and control. We then define the training and inference trigger
supports, introduce the representation quantities required for the
analysis, and state the resulting learning objective.

\noindent\textbf{Learning Setting and Threat Model.}
Let $\cX\subseteq[0,1]^p$ be the input space and
$[C]=\{1,\ldots,C\}$ the label set. For
$m\in\{s,v\}$, the surrogate or victim classifier is written as
\[
    f_m=h_m\circ\phi_m,
\]
where $\phi_m\colon\cX\rightarrow\R^{d_m}$ is the representation map
and $h_m$ outputs class scores
$s_{m,1},\ldots,s_{m,C}$.

The victim trains on $\cD_{\mathrm{tr}}$ using an unknown algorithm
$\mathcal{A}_v$. The attacker may modify at most a fraction $\varrho$
of the training examples, but cannot control the victim architecture,
objective, optimizer, or schedule. The attacker possesses only a
disjoint surrogate dataset and model. This corresponds to a
poison-only, limited-information setting
~\cite{Li2021ISSBA,zeng2023narcissus}, and excludes training-controllable
attacks that jointly optimize the classifier and trigger generator
~\cite{Nguyen2020InputAware}.

A trigger parameter $\theta\in\Theta$ defines a pattern
$g(\theta)=v_\theta$ and a triggered input
\[
    T_\theta(x)=B(x,v_\theta).
\]
Given a target class $t$ and an anchor $\theta_A$, the attacker replaces
examples indexed by $\mathcal I$, where
$|\mathcal I|\leq\lfloor\varrho|\cD_{\mathrm{tr}}|\rfloor$, with
\[
    \cD_{\mathrm{poi}}(\theta_A)
    =
    \left\{
        \bigl(T_{\theta_A}(x_i),t\bigr)
        \,\middle|\,
        i\in\mathcal I
    \right\}.
\]
The victim returns
\[
    f_v^\dagger
    =
    \mathcal A_v
    \left(
        \{(x_i,y_i):i\notin\mathcal I\}
        \cup
        \cD_{\mathrm{poi}}(\theta_A)
    \right).
\]
Only the anchor appears during victim training. The subsequent geometry
also applies to clean-label realizations satisfying the anchor condition. With the victim-side training support fixed, we next formalize the trigger distribution changes at deployment.

\noindent\textbf{Training--Inference Trigger Shift.}
Victim training uses the singleton trigger distribution
\[
    Q_{\mathrm{tr}}=\delta_{\theta_A}.
\]
At inference, the attacker samples from a bounded distribution
$Q_{\mathrm{te}}$ whose support excludes $\theta_A$. Thus, all deployed
family members are absent from victim training. This differs from exact
trigger reuse and from attacks whose trigger diversity is already
training-exposed~\cite{Nguyen2020InputAware,Li2021ISSBA}.

To quantify both exact-anchor behavior and its preservation under this
support shift, let $\cD_{\mathrm{te}}^{\neg t}$ denote clean non-target
test inputs. We measure anchor and family activation as
\begin{align*}
    \asr_A
    &=
    \Prb_{x\sim\cD_{\mathrm{te}}^{\neg t}}
    \bigl[f_v^\dagger(T_{\theta_A}(x))=t\bigr],
    \\
    \asr_F(Q_{\mathrm{te}})
    &=
    \E_{\theta\sim Q_{\mathrm{te}}}
    \Prb_{x\sim\cD_{\mathrm{te}}^{\neg t}}
    \bigl[f_v^\dagger(T_\theta(x))=t\bigr].
\end{align*}
To measure the degradation introduced by the shift, their one-sided
difference defines the $\tgg$ (Trigger Generalization Gap),
\[
    \tgg(Q_{\mathrm{te}})
    =
    \bigl[\asr_A-\asr_F(Q_{\mathrm{te}})\bigr]_+.
\]
A small gap indicates that the anchor-induced behavior remains stable
under trigger shift. To explain when stability is possible, we
next characterize the anchor and inference family geometrically.

\noindent\textbf{Representation Geometry.}
For model $m$, define
\[
    Z_m(x,\theta)=\phi_m(T_\theta(x)),
    \qquad
    \mu_m(\theta)=\E_x[Z_m(x,\theta)].
\]
For failure probability $\delta\in(0,1)$, the concentration radius is
\begin{equation}
\label{eq:radius}
    r_m(\theta,\delta)
    =
    \inf\left\{
        r\geq0
        \,\middle|\,
        \Prb_x
        \left[
            \|Z_m(x,\theta)-\mu_m(\theta)\|_2\leq r
        \right]
        \geq1-\delta
    \right\}.
\end{equation}
We further define
\begin{align*}
    a_m(\theta;\theta_A)
    &=
    \|\mu_m(\theta)-\mu_m(\theta_A)\|_2,
    \\
    \rho_m(\theta,\delta)
    &=
    a_m(\theta;\theta_A)+r_m(\theta,\delta),
    \\
    \gamma_{m,A}
    &=
    s_{m,t}(\mu_m(\theta_A))
    -
    \max_{j\neq t}s_{m,j}(\mu_m(\theta_A)).
\end{align*}
Here, $\gamma_{m,A}$ measures anchor-side separation, while
$\rho_m$ measures how far a variant distribution reaches from the
anchor center. Family activation is possible when the latter remains
within the former.

\noindent\textbf{Analytical Conditions.}
To analyze the preceding geometry and derive sufficient conditions for
family-wise activation, we introduce the following local assumptions.

\begin{assumption}[\textbf{Local Score Regularity}]
\label{ass:score-lipschitz}
For each $m\in\{s,v\}$, every score $s_{m,j}$ is $L_m$-Lipschitz on a
neighborhood containing the representations used in the corresponding
statement.
\end{assumption}

This local condition is standard in margin-based robustness analyses
~\cite{hein2017formal,tsuzuku2018lipschitz,sokolic2017robust} and does
not require a global Lipschitz constant.

\begin{assumption}[\textbf{Learned Anchor Basin}]
\label{ass:anchor-basin}
For some $\delta_A\in(0,1)$, the victim anchor has finite radius
$r_v(\theta_A,\delta_A)$ and positive target margin
$\gamma_{v,A}>0$.
\end{assumption}

Compactness and target separation are desired learned properties, not
universal consequences of poisoning
~\cite{huang2022dbd,qi2023revisiting,Yao2019LatentBackdoor}.

\begin{assumption}[\textbf{Regular Trigger Manifold}]
\label{ass:trigger-manifold}
The pattern map $g$ is continuously differentiable on the admissible
set and satisfies
$\|J_g(\theta)\|_{2\rightarrow2}\leq M_g$.
For fixed $x$, $B(x,\cdot)$ is $L_B$-Lipschitz, and $\phi_m$ is locally
$L_{\phi,m}$-Lipschitz on the corresponding triggered inputs.
\end{assumption}

Bounded parameterizations are common for patch, warp, blending, and
frequency-domain triggers
~\cite{guo2019lowfrequency,Wang2022FrequencyBackdoor}.

Whenever $\gamma_{m,A}>0$, define the normalized geometric risk
\begin{equation}
\label{eq:risk-ratio}
    \zeta_m(\theta,\delta)
    =
    \frac{2L_m\rho_m(\theta,\delta)}{\gamma_{m,A}}.
\end{equation}

\begin{assumption}[\textbf{Bounded Geometry Transfer}]
\label{ass:geometry-transfer}
For each retained family member,
\[
    \zeta_v(\theta,\delta)
    \leq
    \zeta_s(\theta,\delta)+\kappa_{\mathrm{tr}},
\]
where $\kappa_{\mathrm{tr}}\geq0$ bounds surrogate--victim geometric
discrepancy.
\end{assumption}

This is a bounded-transfer abstraction rather than a universal
guarantee. Its plausibility is motivated by cross-model attack transfer
and shared vulnerable directions
~\cite{demontis2019transfer,moosavi2017universal}.

With the trigger shift, representation geometry, and analytical
conditions established, we formalize the resulting attack property.

\begin{definition}[\textbf{Backdoor Generalization under Trigger Shift}]
\label{def:trigger-shift}
An attack exhibits backdoor generalization under trigger shift when
victim training uses only the anchor, deployment uses a non-anchor
inference family, and the resulting model achieves high family ASR with
a small $\tgg$ under the prescribed utility and
stealth budgets.
\end{definition}

\noindent\textbf{Trigger-Shift Learning Objective.}
Definition~\ref{def:trigger-shift} specifies the desired attack
property. We now translate it into the learning objective that a
concrete solution must address. The attacker seeks an anchor and
inference distribution that maximize family activation while satisfying
poisoning, utility, and stealth constraints:
\begin{equation}
\label{eq:attack-objective}
\begin{aligned}
    \max_{\theta_A,Q_{\mathrm{te}}}
    \quad &
    \asr_F(Q_{\mathrm{te}})
    \\
    \text{s.t.}\quad
    &
    Q_{\mathrm{tr}}=\delta_{\theta_A},
    \quad
    \operatorname{supp}(Q_{\mathrm{te}})
    \subseteq\Theta\setminus\{\theta_A\},
    \\
    &
    \varrho\leq\varrho_{\max},
    \quad
    \Delta\caacc\leq\varepsilon_{\mathrm{util}},
    \\
    &
    c_{\mathrm{st}}(\theta_A)\leq\varepsilon_A,
    \quad
    Q_{\mathrm{te}}
    \bigl(c_{\mathrm{st}}(\theta)\leq\varepsilon_F\bigr)=1.
\end{aligned}
\end{equation}
Because the victim model is unknown, the objective cannot be optimized
directly. We therefore try to design a surrogate-guided construction tailored
to its two coupled requirements.

\section{Method}
\label{sec:method}

\subsection{Method Overview}
\label{sec:method-overview}

Equation~\eqref{eq:attack-objective} couples two requirements: creating
a training anchor with sufficient target-side clearance and constructing
an inference family whose representation reach remains within that
clearance. 
Accordingly,
we address
these requirements using disjoint surrogate resources and two sequential
modules. \emph{Module~I, Anchor Vulnerability Induction}, constructs the
sole trigger used in victim training. \emph{Module~II,
Representation-Aligned Family Construction}, fixes the anchor and
constructs a bounded inference-only family that preserves its surrogate
geometry. Both modules operate on the surrogate side; only the anchor
enters victim training, while the family is used after deployment.


Figure~\ref{fig:lilith-geometry} further summarizes the geometric
principle underlying the two modules. Module~I enlarges anchor clearance,
whereas Module~II controls family reach under bounded
surrogate--victim transfer. We next detail their construction and theoretical analysis.

\subsection{Anchor Vulnerability Induction}
\label{sec:anchor}

Module~I searches for an anchor whose surrogate representations are
compact, target-oriented, and budget-compliant. For a candidate
$\theta$, define its target margin and clearance as
\begin{align*}
    \gamma_m(\theta)
    &=
    s_{m,t}(\mu_m(\theta))
    -
    \max_{j\neq t}s_{m,j}(\mu_m(\theta)),
    \\
    \chi_s(\theta,\delta)
    &=
    \frac{\gamma_s(\theta)}{2L_s}
    -
    r_s(\theta,\delta).
\end{align*}
Positive clearance means that the high-probability triggered region lies
inside a target-side margin ball. The implementation uses $L_s$ only
for geometric interpretation and does not estimate a global Lipschitz
constant.

The anchor is selected by
\begin{equation}
\label{eq:anchor-program}
    \theta_A
    \in
    \argmax_{\theta\in\Theta_{\mathrm{adm}}}
    \chi_s(\theta,\delta_A)
    \quad
    \text{s.t.}
    \quad
    c_{\mathrm{st}}(\theta)\leq\varepsilon_A.
\end{equation}
Differentiable estimators increase target consistency, contract
within-trigger representations, and enforce perceptual, amplitude, and
frequency budgets. Together they approximate the single geometric
objective in Eq.~\eqref{eq:anchor-program}.

\begin{proposition}[Anchor basin activation]
\label{prop:anchor-activation}
Suppose Assumptions~\ref{ass:score-lipschitz}
and~\ref{ass:anchor-basin} hold. If
\begin{equation}
\label{eq:anchor-condition}
    \gamma_{v,A}
    >
    2L_v r_v(\theta_A,\delta_A),
\end{equation}
then
\[
    \asr_A\geq1-\delta_A.
\]
\end{proposition}

The proposition shows why Module~I must jointly optimize concentration
and target separation. Either property alone is insufficient to keep
the triggered distribution inside the target region.

\subsection{Representation-Aligned Family Construction}
\label{sec:family}

Module~II fixes $\theta_A$ and constructs an inference distribution over
bounded non-anchor variants. Let
$U\in\R^{q\times k}$ have orthonormal columns spanning an admissible
subspace. Candidate parameters satisfy
\[
    \Theta_A(\eta)
    =
    \left\{
        \theta_A+U\alpha
        \,\middle|\,
        \|\alpha\|_2\leq\eta
    \right\}.
\]
A lightweight generator, random sampler, or quasi-random sampler may
propose $\alpha$. None participates in victim training.

For a distribution $Q$ over $\Theta_A(\eta)$, define its worst-case
surrogate reach as
\[
    \mathcal R_s(Q,\delta)
    =
    \esssup_{\theta\sim Q}
    \rho_s(\theta,\delta).
\]
Module~II maximizes family diversity while retaining only variants that
fit inside the surrogate anchor margin:
\begin{equation}
\label{eq:family-program}
\begin{aligned}
    \max_{Q\in\cQ(\Theta_A(\eta))}
    \quad &
    \mathcal V(Q)
    \\
    \text{s.t.}\quad
    &
    \mathcal R_s(Q,\delta_F)
    <
    \frac{\gamma_{s,A}}{2L_s},
    \\
    &
    Q\bigl(c_{\mathrm{st}}(\theta)\leq\varepsilon_F\bigr)=1.
\end{aligned}
\end{equation}
The implementation combines feature alignment, projection onto the
admissible subspace, and diversity control. Alignment preserves the
anchor-induced representation, projection enforces boundedness, and
diversity prevents collapse.

\begin{table*}[t]
\centering
\caption{\textbf{Attack performance ($\Delta$CA/ASR$_A$/ASR$_F$ $\pm$ Std) under varying $\varrho$}.}

\label{tab:trigger_agnostic_performance}
\begin{adjustbox}{width=0.90\textwidth}
\renewcommand{\arraystretch}{1.2}
\begin{tabular}{cccccccccccccccccc}
\toprule
\multirow{2}{*}{\textbf{Dataset}} & \multirow{2}{*}{\textbf{$\varrho$}} 
& \multicolumn{3}{c}{\textbf{ResNet-18}} 
& \multicolumn{3}{c}{\textbf{ResNet-34}} 
& \multicolumn{3}{c}{\textbf{VGG13-BN}} 
& \multicolumn{3}{c}{\textbf{ViT}} 
& \multicolumn{3}{c}{\textbf{SimpleViT}} \\
\cmidrule(lr){3-5} \cmidrule(lr){6-8} \cmidrule(lr){9-11} \cmidrule(lr){12-14} \cmidrule(lr){15-17}
 &  & \cellcolor{blue!5}$\Delta\caacc$ & \cellcolor{pink!15}$\asr_A$ & \cellcolor{pink!45}$\asr_F$ & \cellcolor{blue!5}$\Delta\caacc$ & \cellcolor{pink!15}$\asr_A$ & \cellcolor{pink!45}$\asr_F$ & \cellcolor{blue!5}$\Delta\caacc$ & \cellcolor{pink!15}$\asr_A$ & \cellcolor{pink!45}$\asr_F$ & \cellcolor{blue!5}$\Delta\caacc$ & \cellcolor{pink!15}$\asr_A$ & \cellcolor{pink!45}$\asr_F$ & \cellcolor{blue!5}$\Delta\caacc$ & \cellcolor{pink!15}$\asr_A$ & \cellcolor{pink!45}$\asr_F$ \\
\midrule
\multirow{3}{*}{\begin{tabular}[c]{@{}c@{}}CIFAR-10\end{tabular}}
 & 0.5\% & \cellcolor{blue!5}1.4\% & \cellcolor{pink!15}97.5\% & \cellcolor{pink!45}92.8\%$\pm$3.2\% & \cellcolor{blue!5}1.6\% & \cellcolor{pink!15}97.4\% & \cellcolor{pink!45}96.1\%$\pm$1.3\% & \cellcolor{blue!5}1.1\% & \cellcolor{pink!15}97.1\% & \cellcolor{pink!45}92.4\%$\pm$3.4\% & \cellcolor{blue!5}1.3\% & \cellcolor{pink!15}96.1\% & \cellcolor{pink!45}93.4\%$\pm$2.3\% & \cellcolor{blue!5}1.6\% & \cellcolor{pink!15}97.0\% & \cellcolor{pink!45}94.2\%$\pm$1.5\% \\
 & 1\% & \cellcolor{blue!5}1.1\% & \cellcolor{pink!15}99.5\% & \cellcolor{pink!45}97.1\%$\pm$1.7\% & \cellcolor{blue!5}1.6\% & \cellcolor{pink!15}99.8\% & \cellcolor{pink!45}99.8\%$\pm$0.2\% & \cellcolor{blue!5}1.2\% & \cellcolor{pink!15}99.6\% & \cellcolor{pink!45}98.6\%$\pm$0.8\% & \cellcolor{blue!5}2.5\% & \cellcolor{pink!15}97.3\% & \cellcolor{pink!45}95.0\%$\pm$1.7\% & \cellcolor{blue!5}1.6\% & \cellcolor{pink!15}97.5\% & \cellcolor{pink!45}95.2\%$\pm$1.2\% \\
 & 5\% & \cellcolor{blue!5}3.4\% & \cellcolor{pink!15}99.9\% & \cellcolor{pink!45}100.0\%$\pm$0.1\% & \cellcolor{blue!5}3.1\% & \cellcolor{pink!15}100.0\% & \cellcolor{pink!45}100.0\%$\pm$0.1\% & \cellcolor{blue!5}1.6\% & \cellcolor{pink!15}99.9\% & \cellcolor{pink!45}99.9\%$\pm$0.3\% & \cellcolor{blue!5}1.2\% & \cellcolor{pink!15}99.6\% & \cellcolor{pink!45}98.1\%$\pm$1.0\% & \cellcolor{blue!5}1.1\% & \cellcolor{pink!15}99.4\% & \cellcolor{pink!45}98.3\%$\pm$1.2\% \\
\midrule
\multirow{3}{*}{\begin{tabular}[c]{@{}c@{}}CIFAR-100\end{tabular}}
 & 0.5\% & \cellcolor{blue!5}1.2\% & \cellcolor{pink!15}93.6\% & \cellcolor{pink!45}91.7\%$\pm$3.2\% & \cellcolor{blue!5}2.7\% & \cellcolor{pink!15}95.9\% & \cellcolor{pink!45}94.2\%$\pm$1.7\% & \cellcolor{blue!5}2.6\% & \cellcolor{pink!15}97.3\% & \cellcolor{pink!45}92.3\%$\pm$3.0\% & \cellcolor{blue!5}1.3\% & \cellcolor{pink!15}95.6\% & \cellcolor{pink!45}91.6\%$\pm$2.9\% & \cellcolor{blue!5}1.4\% & \cellcolor{pink!15}94.1\% & \cellcolor{pink!45}92.2\%$\pm$2.0\% \\
 & 1\% & \cellcolor{blue!5}1.1\% & \cellcolor{pink!15}97.2\% & \cellcolor{pink!45}94.1\%$\pm$1.9\% & \cellcolor{blue!5}3.1\% & \cellcolor{pink!15}98.2\% & \cellcolor{pink!45}97.7\%$\pm$1.3\% & \cellcolor{blue!5}3.4\% & \cellcolor{pink!15}97.2\% & \cellcolor{pink!45}95.5\%$\pm$2.7\% & \cellcolor{blue!5}1.4\% & \cellcolor{pink!15}96.4\% & \cellcolor{pink!45}93.0\%$\pm$2.8\% & \cellcolor{blue!5}1.9\% & \cellcolor{pink!15}95.4\% & \cellcolor{pink!45}92.9\%$\pm$2.7\% \\
 & 5\% & \cellcolor{blue!5}2.3\% & \cellcolor{pink!15}99.9\% & \cellcolor{pink!45}99.4\%$\pm$0.6\% & \cellcolor{blue!5}4.1\% & \cellcolor{pink!15}99.8\% & \cellcolor{pink!45}99.8\%$\pm$0.2\% & \cellcolor{blue!5}1.2\% & \cellcolor{pink!15}99.9\% & \cellcolor{pink!45}99.6\%$\pm$0.5\% & \cellcolor{blue!5}1.8\% & \cellcolor{pink!15}98.4\% & \cellcolor{pink!45}98.4\%$\pm$1.2\% & \cellcolor{blue!5}1.0\% & \cellcolor{pink!15}98.7\% & \cellcolor{pink!45}98.3\%$\pm$1.2\% \\
\midrule
\multirow{3}{*}{\begin{tabular}[c]{@{}c@{}}TinyImageNet\end{tabular}}
 & 0.5\% & \cellcolor{blue!5}2.8\% & \cellcolor{pink!15}93.9\% & \cellcolor{pink!45}91.5\%$\pm$3.2\% & \cellcolor{blue!5}3.3\% & \cellcolor{pink!15}94.3\% & \cellcolor{pink!45}92.5\%$\pm$2.3\% & \cellcolor{blue!5}3.4\% & \cellcolor{pink!15}95.8\% & \cellcolor{pink!45}91.5\%$\pm$3.6\% & \cellcolor{blue!5}0.8\% & \cellcolor{pink!15}95.8\% & \cellcolor{pink!45}93.6\%$\pm$2.2\% & \cellcolor{blue!5}0.5\% & \cellcolor{pink!15}95.9\% & \cellcolor{pink!45}92.6\%$\pm$2.3\% \\
 & 1\% & \cellcolor{blue!5}2.9\% & \cellcolor{pink!15}99.7\% & \cellcolor{pink!45}98.9\%$\pm$0.6\% & \cellcolor{blue!5}4.3\% & \cellcolor{pink!15}99.8\% & \cellcolor{pink!45}99.4\%$\pm$0.3\% & \cellcolor{blue!5}4.5\% & \cellcolor{pink!15}99.9\% & \cellcolor{pink!45}99.1\%$\pm$0.6\% & \cellcolor{blue!5}0.4\% & \cellcolor{pink!15}95.3\% & \cellcolor{pink!45}94.9\%$\pm$1.5\% & \cellcolor{blue!5}1.0\% & \cellcolor{pink!15}96.8\% & \cellcolor{pink!45}93.3\%$\pm$1.3\% \\
 & 5\% & \cellcolor{blue!5}3.7\% & \cellcolor{pink!15}100.0\% & \cellcolor{pink!45}100.0\%$\pm$0.1\% & \cellcolor{blue!5}5.0\% & \cellcolor{pink!15}100.0\% & \cellcolor{pink!45}100.0\%$\pm$0.0\% & \cellcolor{blue!5}5.3\% & \cellcolor{pink!15}100.0\% & \cellcolor{pink!45}100.0\%$\pm$0.1\% & \cellcolor{blue!5}1.4\% & \cellcolor{pink!15}99.0\% & \cellcolor{pink!45}95.1\%$\pm$1.7\% & \cellcolor{blue!5}1.1\% & \cellcolor{pink!15}97.1\% & \cellcolor{pink!45}97.7\%$\pm$1.4\% \\
\midrule
\multirow{3}{*}{\begin{tabular}[c]{@{}c@{}}SVHN\end{tabular}}
 & 0.5\% & \cellcolor{blue!5}0.8\% & \cellcolor{pink!15}95.6\% & \cellcolor{pink!45}90.2\%$\pm$4.7\% & \cellcolor{blue!5}1.0\% & \cellcolor{pink!15}98.1\% & \cellcolor{pink!45}95.3\%$\pm$1.2\% & \cellcolor{blue!5}0.8\% & \cellcolor{pink!15}96.4\% & \cellcolor{pink!45}91.7\%$\pm$2.9\% & \cellcolor{blue!5}0.9\% & \cellcolor{pink!15}94.2\% & \cellcolor{pink!45}90.0\%$\pm$4.0\% & \cellcolor{blue!5}0.8\% & \cellcolor{pink!15}96.2\% & \cellcolor{pink!45}93.5\%$\pm$1.9\% \\
 & 1\% & \cellcolor{blue!5}1.0\% & \cellcolor{pink!15}98.5\% & \cellcolor{pink!45}98.1\%$\pm$1.1\% & \cellcolor{blue!5}1.5\% & \cellcolor{pink!15}99.9\% & \cellcolor{pink!45}99.5\%$\pm$0.1\% & \cellcolor{blue!5}1.5\% & \cellcolor{pink!15}100.0\% & \cellcolor{pink!45}99.8\%$\pm$0.1\% & \cellcolor{blue!5}0.6\% & \cellcolor{pink!15}98.7\% & \cellcolor{pink!45}97.5\%$\pm$0.9\% & \cellcolor{blue!5}1.0\% & \cellcolor{pink!15}95.8\% & \cellcolor{pink!45}94.0\%$\pm$1.7\% \\
 & 5\% & \cellcolor{blue!5}1.2\% & \cellcolor{pink!15}100.0\% & \cellcolor{pink!45}100.0\%$\pm$0.0\% & \cellcolor{blue!5}2.1\% & \cellcolor{pink!15}100.0\% & \cellcolor{pink!45}100.0\%$\pm$0.0\% & \cellcolor{blue!5}2.3\% & \cellcolor{pink!15}100.0\% & \cellcolor{pink!45}100.0\%$\pm$0.1\% & \cellcolor{blue!5}1.8\% & \cellcolor{pink!15}98.8\% & \cellcolor{pink!45}94.3\%$\pm$2.3\% & \cellcolor{blue!5}1.7\% & \cellcolor{pink!15}99.5\% & \cellcolor{pink!45}98.2\%$\pm$0.6\% \\
\midrule
\multirow{3}{*}{\begin{tabular}[c]{@{}c@{}}ImageNet-1K\end{tabular}}
 & 0.5\% & \cellcolor{blue!5}2.5\% & \cellcolor{pink!15}99.0\% & \cellcolor{pink!45}98.9\%$\pm$1.9\% & \cellcolor{blue!5}2.9\% & \cellcolor{pink!15}98.9\% & \cellcolor{pink!45}97.5\%$\pm$1.0\% & \cellcolor{blue!5}2.7\% & \cellcolor{pink!15}98.9\% & \cellcolor{pink!45}96.3\%$\pm$2.0\% & \cellcolor{blue!5}0.6\% & \cellcolor{pink!15}98.5\% & \cellcolor{pink!45}96.0\%$\pm$1.7\% & \cellcolor{blue!5}0.3\% & \cellcolor{pink!15}97.4\% & \cellcolor{pink!45}95.1\%$\pm$2.0\% \\
 & 1\% & \cellcolor{blue!5}2.5\% & \cellcolor{pink!15}100.0\% & \cellcolor{pink!45}99.8\%$\pm$0.4\% & \cellcolor{blue!5}3.4\% & \cellcolor{pink!15}100.0\% & \cellcolor{pink!45}99.9\%$\pm$0.1\% & \cellcolor{blue!5}2.9\% & \cellcolor{pink!15}99.9\% & \cellcolor{pink!45}99.5\%$\pm$0.2\% & \cellcolor{blue!5}1.0\% & \cellcolor{pink!15}98.9\% & \cellcolor{pink!45}97.0\%$\pm$0.5\% & \cellcolor{blue!5}0.8\% & \cellcolor{pink!15}98.2\% & \cellcolor{pink!45}97.0\%$\pm$0.9\% \\
 & 5\% & \cellcolor{blue!5}3.1\% & \cellcolor{pink!15}100.0\% & \cellcolor{pink!45}100.0\%$\pm$0.0\% & \cellcolor{blue!5}3.5\% & \cellcolor{pink!15}100.0\% & \cellcolor{pink!45}100.0\%$\pm$0.1\% & \cellcolor{blue!5}4.7\% & \cellcolor{pink!15}100.0\% & \cellcolor{pink!45}100.0\%$\pm$0.0\% & \cellcolor{blue!5}1.5\% & \cellcolor{pink!15}99.6\% & \cellcolor{pink!45}98.7\%$\pm$0.6\% & \cellcolor{blue!5}0.8\% & \cellcolor{pink!15}99.6\% & \cellcolor{pink!45}98.0\%$\pm$0.8\% \\
\bottomrule
\end{tabular}
\end{adjustbox}
\vspace{-1em}
\end{table*}

\begin{proposition}[\textbf{Parameter-to-representation Stability}]
\label{prop:parameter-stability}
Suppose Assumption~\ref{ass:trigger-manifold} holds for model $m$, and
let
\[
    C_m=L_{\phi,m}L_BM_g.
\]
For any $\theta\in\Theta_A(\eta)$ and admissible $x$,
\begin{equation}
\label{eq:pointwise-stability}
    \|Z_m(x,\theta)-Z_m(x,\theta_A)\|_2
    \leq C_m\eta.
\end{equation}
Moreover,
\begin{align}
\label{eq:center-bound}
    a_m(\theta;\theta_A)
    &\leq C_m\eta,
    \\
\label{eq:radius-bound}
    r_m(\theta,\delta_A)
    &\leq r_m(\theta_A,\delta_A)+2C_m\eta,
    \\
\label{eq:reach-bound}
    \rho_m(\theta,\delta_A)
    &\leq r_m(\theta_A,\delta_A)+3C_m\eta.
\end{align}
\end{proposition}

This conservative result explains why a bounded smooth manifold is a
suitable proposal space. Direct feature screening may retain additional
variants whose observed reach remains small.

\begin{theorem}[\textbf{Distributional Family Activation}]
\label{thm:family-activation}
Suppose Assumption~\ref{ass:score-lipschitz} holds and
$\gamma_{v,A}>0$. If member $\theta$ has concentration failure
probability $\delta(\theta)$, then
\begin{equation}
\label{eq:family-lower-bound}
    \asr_F(Q_{\mathrm{te}})
    \geq
    \E_{\theta\sim Q_{\mathrm{te}}}
    \left[
        \bigl(1-\delta(\theta)\bigr)
        \ind\left\{
            \zeta_v(\theta,\delta(\theta))<1
        \right\}
    \right].
\end{equation}
\end{theorem}

A family member succeeds when its representation reach remains below
the anchor clearance. Pixel-level similarity is neither required nor
sufficient.

\begin{corollary}[Radius-controlled family activation]
\label{cor:radius-family}
Suppose Assumptions~\ref{ass:score-lipschitz}
and~\ref{ass:trigger-manifold} hold for the victim. If
\begin{equation}
\label{eq:radius-family-condition}
    r_v(\theta_A,\delta_A)+3C_v\eta
    <
    \frac{\gamma_{v,A}}{2L_v},
\end{equation}
then every distribution supported on $\Theta_A(\eta)$ satisfies
\[
    \asr_F(Q_{\mathrm{te}})\geq1-\delta_A.
\]
\end{corollary}

\begin{corollary}[Uniform family guarantee]
\label{cor:uniform-family}
If, for some $q<1$ and $\bar\delta\in(0,1)$,
\[
    \esssup_{\theta\sim Q_{\mathrm{te}}}
    \zeta_v(\theta,\bar\delta)
    \leq q,
\]
then
\[
    \asr_F(Q_{\mathrm{te}})\geq1-\bar\delta,
    \qquad
    \tgg(Q_{\mathrm{te}})\leq\bar\delta.
\]
\end{corollary}

\begin{corollary}[\textbf{Surrogate Screening under Bounded Transfer}]
\label{cor:surrogate-transfer}
Suppose Assumption~\ref{ass:geometry-transfer} holds with
$\kappa_{\mathrm{tr}}<1$. If
\begin{equation}
\label{eq:surrogate-screening}
    \esssup_{\theta\sim Q_{\mathrm{te}}}
    \zeta_s(\theta,\bar\delta)
    <
    1-\kappa_{\mathrm{tr}},
\end{equation}
then
\[
    \asr_F(Q_{\mathrm{te}})\geq1-\bar\delta.
\]
\end{corollary}

The final corollary connects surrogate screening to black-box victim
activation. Its bounded-transfer premise is evaluated through model,
dataset, and proposal-mechanism mismatch.

\noindent\textbf{Scope.}
Our analysis identifies three geometric factors governing family-wise
activation: anchor clearance, family reach, and surrogate--victim
transfer. The resulting sufficient conditions provide concrete design
criteria for inducing the anchor vulnerability and screening
inference-time variants, while explaining why geometry-aligned
in-subspace variants remain effective and off-subspace variants fail.
Rather than modeling the full non-convex training dynamics, the analysis
isolates the local representation geometry that a successful
anchor-to-family attack must establish and that our following experiments examine
across victim settings.

\section{Evaluation}

Our evaluation addresses five questions derived from
Definition~\ref{def:trigger-shift} and the design of \textbf{Lilith}:

\noindent$\bullet$\quad\textbf{RQ1.}
Can a singleton training trigger achieve high family-wise activation
with a small $\tgg$ across datasets, victim
architectures, and poisoning budgets?

\noindent$\bullet$\quad\textbf{RQ2.}
Does \textbf{Lilith} preserve malicious behavior under matched
training--inference trigger shift more effectively than prior attacks?

\noindent$\bullet$\quad\textbf{RQ3.}
Does family-wise activation depend on representation alignment rather
than a learned generator?

\noindent$\bullet$\quad\textbf{RQ4.}
How perceptually stealthy and detectable is \textbf{Lilith} under the evaluated
input-, prediction-, and feature-level detectors?

\noindent$\bullet$\quad\textbf{RQ5.}
How stable is \textbf{Lilith} under deployment perturbations, defenses, and
resource changes?

\begin{figure}[t]
    \centering
    \begin{subfigure}[t]{0.47\linewidth}
        \centering
        \includegraphics[width=\linewidth]{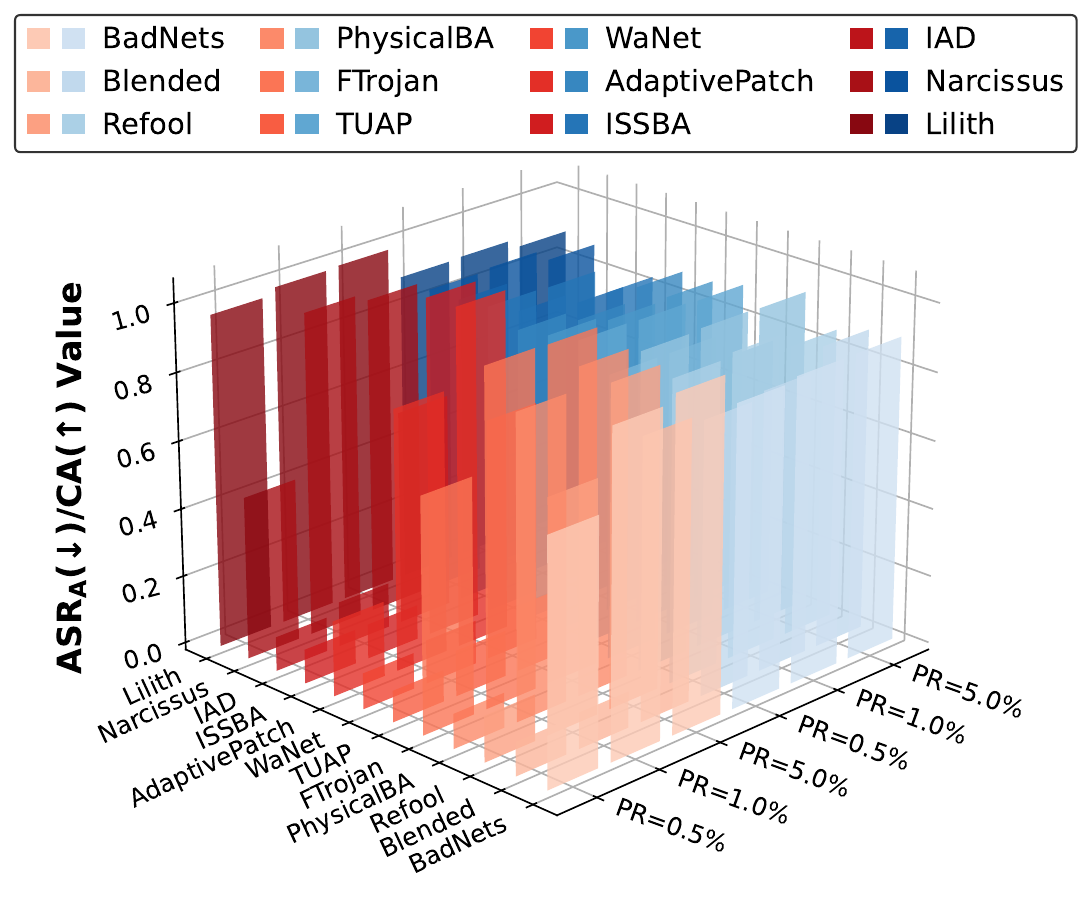}
        \caption{Comparison under trigger-specific setting.}
        \label{fig:attacks_comparison}
    \end{subfigure}
    \begin{subfigure}[t]{0.47\linewidth}
        \centering
        \includegraphics[width=\linewidth]{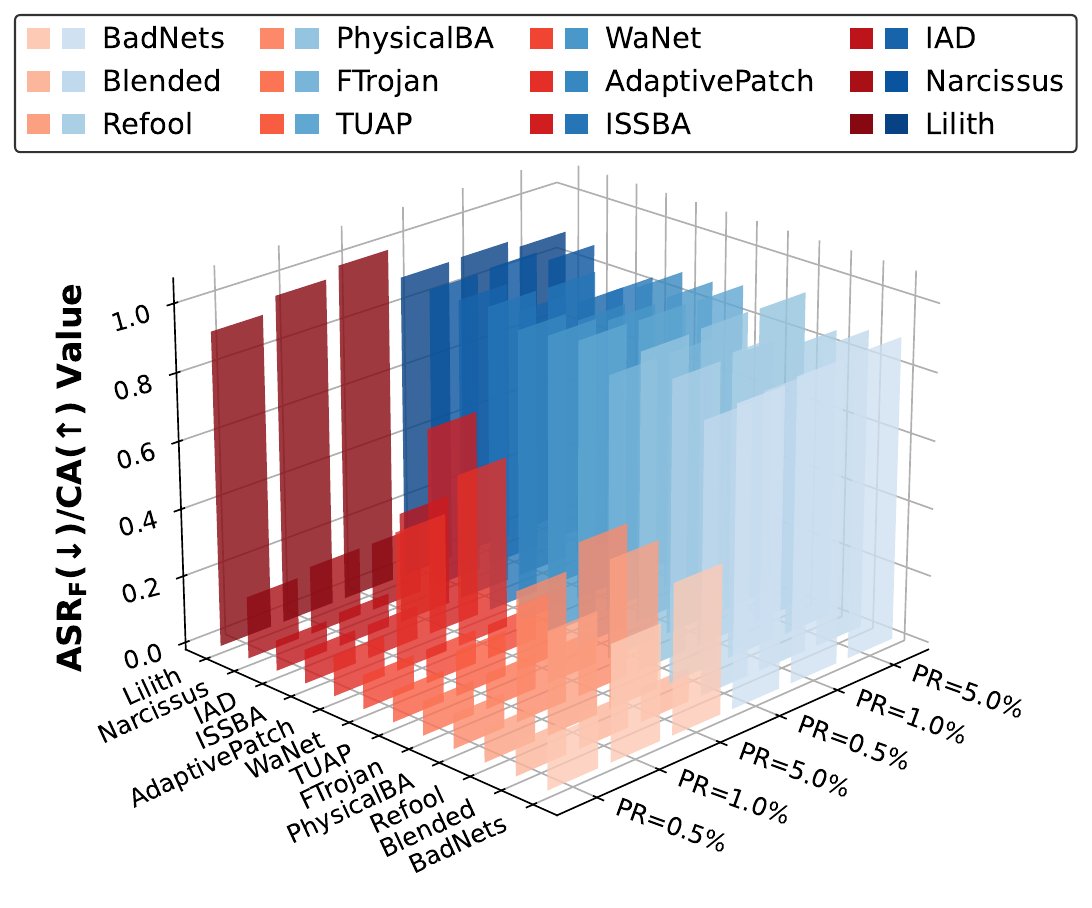}
        \caption{Comparison under trigger-agnostic setting.}
        \label{fig:attacks_comparison_noised}
    \end{subfigure}
    \caption{Comparison of \textbf{Lilith} with SOTA backdoor attacks.}
    \label{fig:total_comparison}
    \vspace{-1em}
\end{figure}

\subsection{Experimental Setup}

\noindent\textbf{Datasets and Models.}
We evaluate on CIFAR-10, CIFAR-100~\cite{Krizhevsky2009CIFAR10}, TinyImageNet~\cite{Le2015TinyImageNet}, SVHN~\cite{Goodfellow2014SVHN} and ImageNet-1K~\cite{Deng2009ImageNet}.
The victim architectures include ResNet-18, ResNet-34~\cite{He2016ResNet}, VGG13-BN~\cite{Simonyan2015VGG}, ViT~\cite{Dosovitskiy2021ViT}, and SimpleViT~\cite{Beyer2022PlainViT}.
Unless stated otherwise, the default setting uses a ResNet-18 victim on CIFAR-10, poisoning rate $\varrho=1\%$, attacker target $t=0$, and dirty-label poisoning.
The surrogate data and architecture
are disjoint from those of the victim, following the poison-only
black-box setting in Section~\ref{sec:problem}.

\noindent\textbf{Attack Protocol.}
We follow the training--inference trigger-shift protocol and metric
definitions in Section~\ref{sec:problem}. For matched attack
comparisons, each baseline is trained with its canonical trigger and
evaluated on an unseen inference family constructed under the same
transformation set and perceptual budget. This protocol measures
behavioral retention under trigger-support shift rather than canonical
trigger reuse. Table~\ref{tab:metrics} summarizes the reported metrics.

\begin{table}[t]
\small
\centering
\caption{\textbf{Evaluation metrics for backdoor generalization under trigger shift.}}
\label{tab:metrics}
\begin{adjustbox}{width=1.0\linewidth}
\begin{tabularx}{\linewidth}{lX}
\toprule
\textbf{Metric} & \textbf{Definition} \\
\midrule

$\caacc$ &
Clean accuracy on benign test examples. \\

$\Delta\caacc$ &
Clean-accuracy decrease relative to a benign model trained with the same configuration. \\

$\asr_A$ &
Anchor activation $\asr(\theta_A)$. It measures exact reuse of the only trigger exposed during victim training. \\

$\asr_F(Q_{\mathrm{te}})$ &
Family activation $\E_{\theta\sim Q_{\mathrm{te}}}[\asr(\theta)]$. It is the primary measure of unseen-family activation. \\

$\tgg(Q_{\mathrm{te}})$ &
$[\asr_A-\asr_F(Q_{\mathrm{te}})]_+$, one-sided trigger generalization gap. Smaller values indicate stronger retention under trigger shift. \\

$\sigma_F(Q_{\mathrm{te}})$ &
Variant-wise dispersion $\operatorname{Std}_{\theta\sim Q_{\mathrm{te}}}[\asr(\theta)]$. It measures whether success is concentrated on only a few variants. \\

$\varrho$ &
Poisoning rate, defined as the fraction of victim training examples replaced by poisoned examples. \\

\bottomrule
\end{tabularx}
\end{adjustbox}
\vspace{-1em}
\end{table}

\subsection{RQ1. Effectiveness under Trigger Shift}
\label{sec:eval-rq1}
Table~\ref{tab:trigger_agnostic_performance} reports $\Delta\caacc$,
$\asr_A$, and $\asr_F$ across five datasets, five victim architectures,
and three poisoning rates. At $\varrho=0.5\%$, $\asr_F$ reaches at
least $90.0\%$ in all 25 dataset--architecture pairs, with
$\Delta\caacc$ no greater than $3.4\%$. At $\varrho=1\%$,
$\asr_F$ ranges from $92.9\%$ to $99.9\%$ and remains within $3.5$
percentage points of $\asr_A$. At $\varrho=5\%$, most CNN settings
approach $100\%$ family activation, while all transformer settings
remain above $94\%$. Across all 75 configurations, the derived
$\tgg$ remains single-digit, and the generally
small standard deviations indicate stable activation across sampled
variants. These results show that one training anchor consistently
supports family-wise activation across data scales and model families.

\begin{figure}
    \centering
    \includegraphics[width=0.80\linewidth]{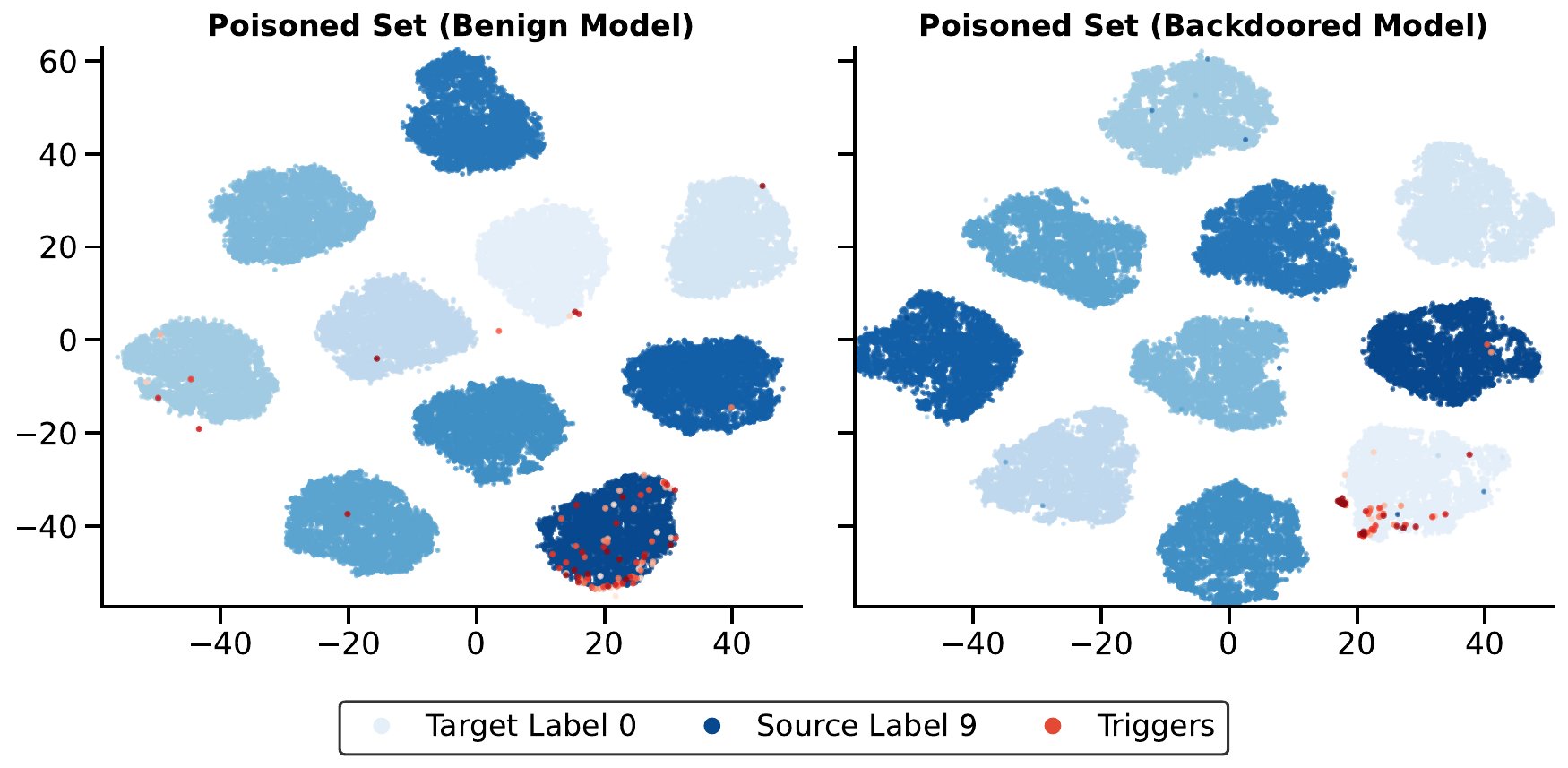}
    \caption{T-SNE of poisoned representations.}
    \label{fig:representation}
    \vspace{-1em}
\end{figure}

\subsection{RQ2. Comparison under Trigger Shift}
\label{sec:eval-rq2}
Figure~\ref{fig:total_comparison} compares \textbf{Lilith} with representative
fixed-pattern and adaptive attacks
~\cite{gu2017badnets,liu2020reflection,qi2023revisiting},
structured frequency- or transformation-based attacks
~\cite{guo2019lowfrequency,li2021physical,nguyen2021wanet},
clean-label attacks~\cite{zeng2023narcissus,zhao2020clean},
and dynamic or sample-specific attacks
~\cite{Nguyen2020InputAware,Li2021ISSBA}.
We report anchor reuse through $\asr_A$ and matched trigger-shift
performance through $\asr_F$, where each attack is evaluated on an
unseen family under the same perturbation and transformation budget. Most baselines achieve high anchor-reuse success but suffer larger drops
under trigger shift. \textbf{Lilith} retains more of its anchor activation and
therefore yields a smaller $\tgg$ across poisoning rates, demonstrating
stronger singleton-to-family generalization under the matched protocol.
Dynamic and sample-specific attacks provide useful contextual
comparisons, although their trigger diversity is already exposed during
training or relies on stronger training control. We account for this
difference when interpreting their results.

\begin{figure}[]
    \centering
    \begin{subfigure}[b]{0.20\textwidth}
        \centering
        \includegraphics[width=\textwidth]{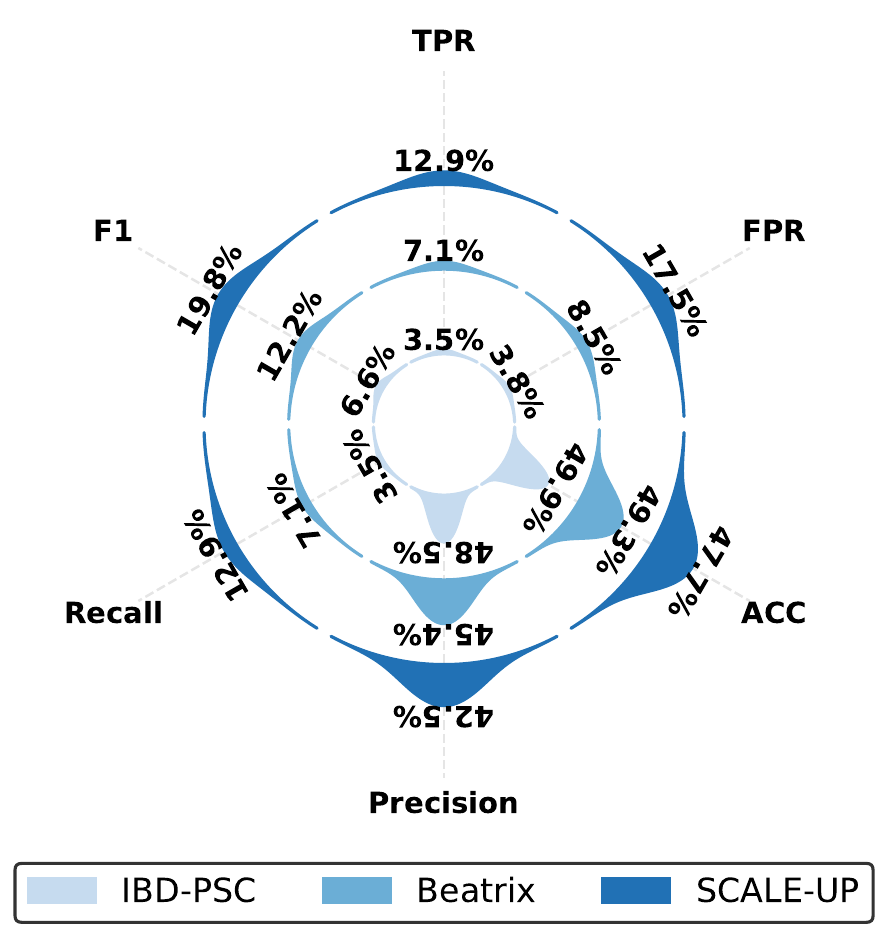}
        \caption{Training-stage Defenses}
        \label{fig:radar_detection}
    \end{subfigure}
    \begin{subfigure}[b]{0.20\textwidth}
        \centering
        \includegraphics[width=\textwidth]{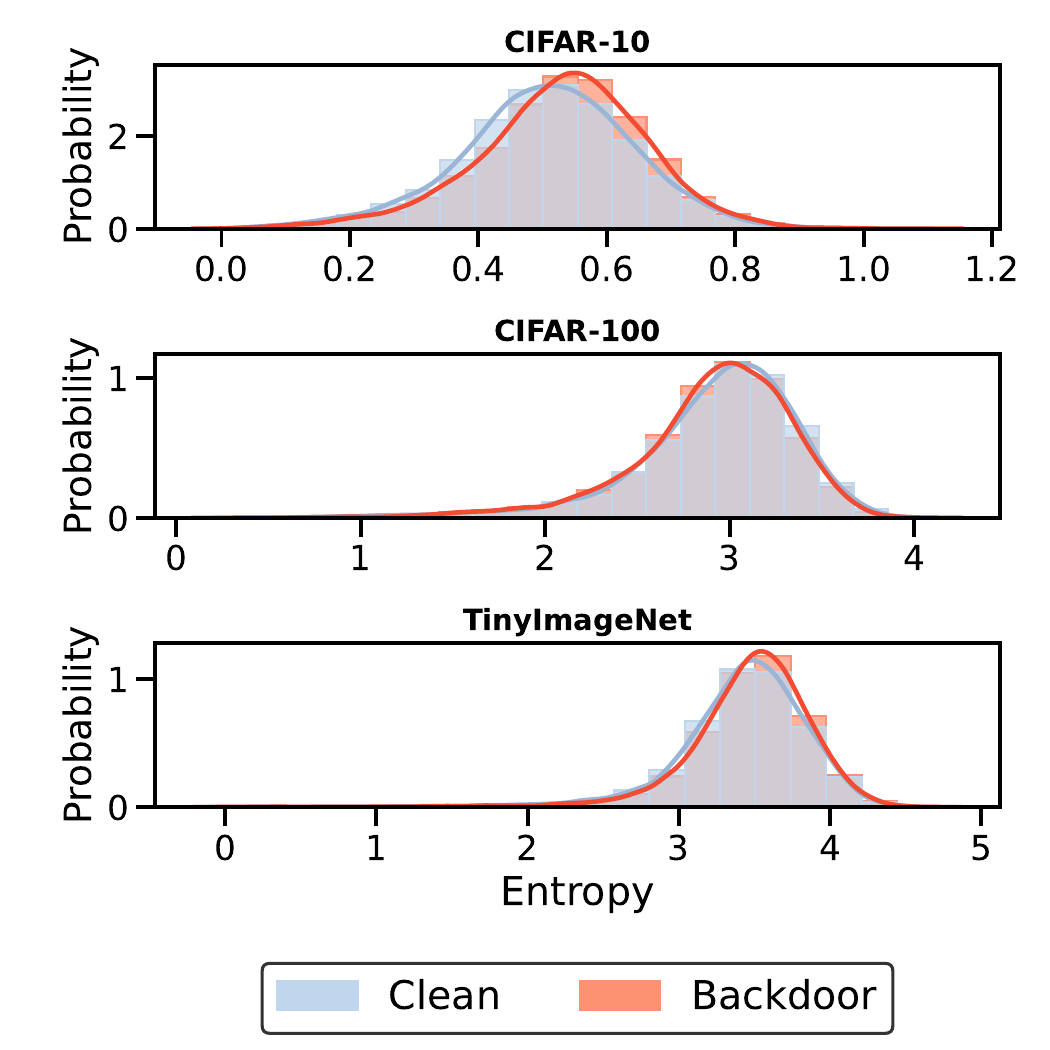}
        \caption{Entropy Distributions}
        \label{fig:strip}
    \end{subfigure}
    \hfill
    \begin{subfigure}[b]{0.40\textwidth}
        \centering
        \includegraphics[width=\textwidth]{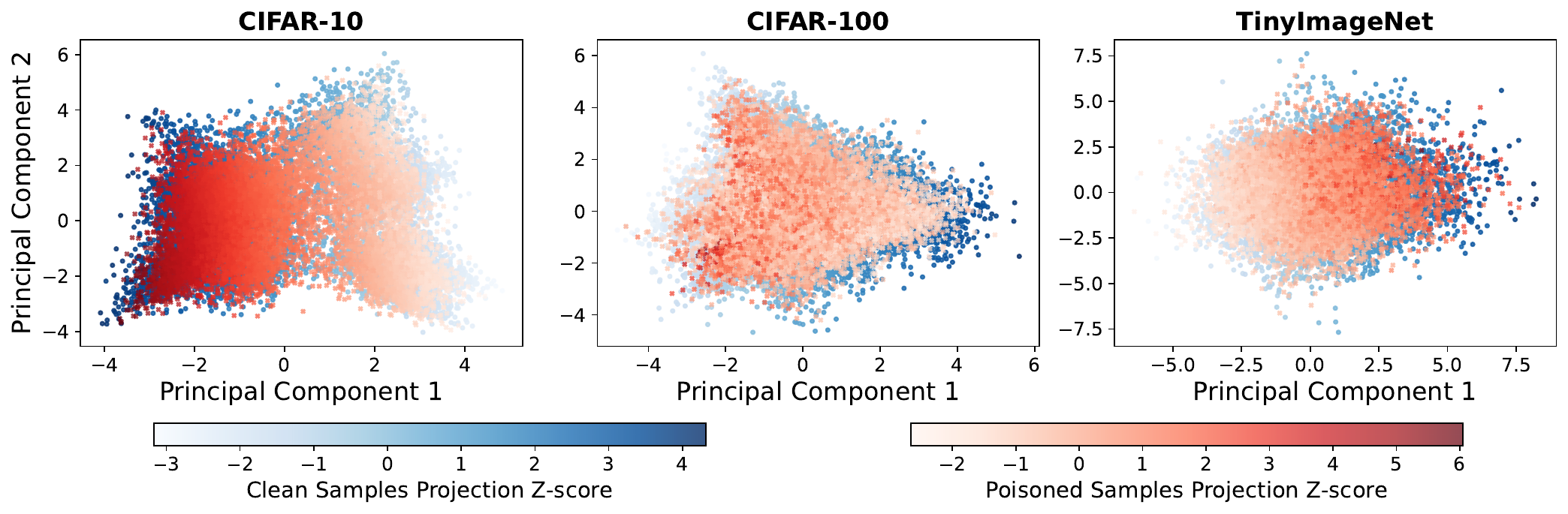}
        \caption{PCA Projections of Spectral~\cite{Tran2018SpectralSignatures} Statistics}
        \label{fig:spectral_PCA}
    \end{subfigure}
    \caption{Evaluation of \textbf{Lilith}’s resistance to detections.}
    \label{fig:stealthiness_study}
    \vspace{-1em}
\end{figure}

\subsection{RQ3. Representation Alignment}
\label{sec:eval-rq3}

\begin{table}[t]
\centering
\footnotesize
\setlength{\tabcolsep}{2.3pt}
\caption{Proposal-source and geometry-alignment stress test.}
\label{tab:family_source_stress}
\resizebox{0.95\linewidth}{!}{%
\begin{tabular}{@{}lcccc@{}}
\toprule
\textbf{Variant Distribution}
&
\textbf{Train-visible}
&
\textbf{Alignment Mechanism}
&
\textbf{Activation}
&
\textbf{TGG} $\downarrow$
\\
\midrule

\rowcolor{pink!15}
Anchor Reference
$Q_A=\delta_{\theta_A}$
&
\cmark
&
Module I
&
$\mathrm{ASR}_{A}=99.5\%\pm0.0\%$
&
---
\\

\rowcolor{pink!45}
Learned Proposal
$Q_{\mathrm{gen}}$
&
\xmark
&
Learned Alignment
&
$\mathrm{ASR}_{F}=97.1\%\pm1.7\%$
&
$2.4\%$
\\

Random In-subspace
$Q_{\mathrm{rand}}$
&
\xmark
&
None
&
$\mathrm{ASR}_{F}=75.0\%\pm9.9\%$
&
$24.5\%$
\\

Feature-filtered Random
$Q_{\mathrm{filt}}$
&
\xmark
&
Geometry Screening
&
$\mathrm{ASR}_{F}=93.5\%\pm2.4\%$
&
$6.0\%$
\\

Sobol In-subspace
$Q_{\mathrm{sobol}}$
&
\xmark
&
Bounded Sampling
&
$\mathrm{ASR}_{F}=91.4\%\pm1.6\%$
&
$8.1\%$
\\

Off-subspace
$Q_{\mathrm{off}}$
&
\xmark
&
None
&
$\mathrm{ASR}_{F}=10.1\%\pm0.1\%$
&
$89.4\%$
\\

\bottomrule
\end{tabular}
}
\vspace{-0.8em}
\end{table}
Table~\ref{tab:family_source_stress} separates candidate generation
from geometric screening. The learned proposal achieves $97.1\%$
family activation. Random Fourier offsets remain partially effective,
while feature-filtered random and Sobol offsets exceed $90\%$ without
generator optimization. Off-subspace offsets instead approach chance
performance. These results show that diversity alone is insufficient, where
effective variants must preserve the surrogate anchor geometry.

Figure~\ref{fig:representation} provides consistent qualitative
evidence. Under the backdoored model, anchor- and family-triggered
examples concentrate near a target-associated region, supporting the
anchor-basin interpretation in Section~\ref{sec:method}. Together, the
results indicate that representation alignment, rather than dependence
on one proposal distribution, primarily governs family activation.

\subsection{RQ4. Perceptual Stealth and Detectability}
\label{sec:eval-rq4}

\begin{figure}[]
    \centering
    \begin{subfigure}[b]{0.06\textwidth}
        \centering
        \includegraphics[width=\textwidth]{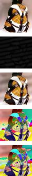}
        \caption{}
        \label{fig:trigger_0}
    \end{subfigure}
    \hfill
    \begin{subfigure}[b]{0.06\textwidth}
        \centering
        \includegraphics[width=\textwidth]{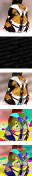}
        \caption{}
        \label{fig:trigger_1}
    \end{subfigure}
    \hfill
    \begin{subfigure}[b]{0.06\textwidth}
        \centering
        \includegraphics[width=\textwidth]{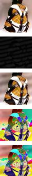}
        \caption{}
        \label{fig:trigger_2}
    \end{subfigure}
    \hfill
    \begin{subfigure}[b]{0.06\textwidth}
        \centering
        \includegraphics[width=\textwidth]{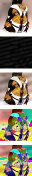}
        \caption{}
        \label{fig:trigger_3}
    \end{subfigure}
    \hfill
    \begin{subfigure}[b]{0.06\textwidth}
        \centering
        \includegraphics[width=\textwidth]{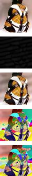}
        \caption{}
        \label{fig:trigger_4}
    \end{subfigure}
    \hfill
    \begin{subfigure}[b]{0.06\textwidth}
        \centering
        \includegraphics[width=\textwidth]{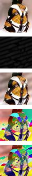}
        \caption{}
        \label{fig:trigger_5}
    \end{subfigure}
    \caption{Visualization of six triggers (a)--(f) for the stealthiness study. The \textbf{first} and \textbf{third rows} show the original and triggered images, while the \textbf{second} shows their residuals. The \textbf{fourth} and \textbf{fifth rows} present Grad-CAM~\cite{Selvaraju2017GradCAM} heatmaps from a benign model for the original and triggered inputs. The triggers introduce low distortion, with an \textbf{average MSE of 0.0033}, \textbf{PSNR of 24.76 dB}, 
    \textbf{$\ell_{\infty}$-norm of 0.1057} and \textbf{LPIPS~\cite{Zhang2018PerceptualMetric} of 0.0335}
    }

    \label{fig:triggers_visual_study}
    \vspace{-1em}
\end{figure}

\noindent\textbf{Perceptual Behavior.}
Figure~\ref{fig:triggers_visual_study} shows clean images, triggered images, residuals, and Grad-CAM responses for multiple family members.
The perturbations remain visually subtle under the reported pixel- and perceptual-distance measures.
The attention responses also vary across variants rather than collapsing to one repeated salient location.
This evidence supports perceptual diversity within the retained family.

\noindent\textbf{Resilience to Detection.}
We evaluate prediction-consistency, fea-ture-statistics, entropy-based, and spectral detectors.
IBD-PSC~\cite{hou2024ibdpsc}, SCALE-UP~\cite{guo2023scaleup}, and Beatrix~\cite{ma2023beatrix} obtain detection accuracy near chance with low recall in the default setting.
STRIP~\cite{gao2019strip} produces strongly overlapping entropy distributions for clean and triggered inputs, while the spectral-signature projections~\cite{Tran2018SpectralSignatures} in Figure~\ref{fig:spectral_PCA} show limited separation between clean and poisoned examples.
These results show that the evaluated detectors do not reliably identify \textbf{Lilith}under the default protocol.

\subsection{RQ5. Robustness, Mitigation and Sensitivity}
\label{sec:eval-rq5}

\noindent\textbf{Deployment Perturbations.}
Figure~\ref{fig:augmentation} evaluates horizontal flipping, Gaussian noise, JPEG compression, resizing, occlusion, and brightness or contrast changes.
Family activation remains high across these transformations with little change in clean accuracy.
This indicates that the retained family is not tied to exact pixel reproduction under the evaluated preprocessing operations, which is consistent with representation alignment.

\noindent\textbf{Mitigation and Defense Granularity.}
Figure~\ref{fig:mitigation} and the extended evaluation cover input
purification~\cite{Yang2024SampDetox}, representation-centric training
and cleansing~\cite{huang2022dbd,Min2023FST,li2021antibackdoor,
li2021neural}, and feature repair
~\cite{Zheng2024SSLCleanse,zhu2023enhancing}.
Input-level methods generally leave higher residual $\asr_F$ with
limited utility loss, whereas DBD, SSL-Cleanse, and FST suppress family
activation more strongly by modifying feature geometry or
trigger--label associations, often at a larger $\Delta\caacc$.
We therefore report attack suppression together with clean utility.

\begin{figure}[]
    \centering
    \begin{subfigure}[b]{0.235\textwidth}
        \centering
        \includegraphics[width=\textwidth]{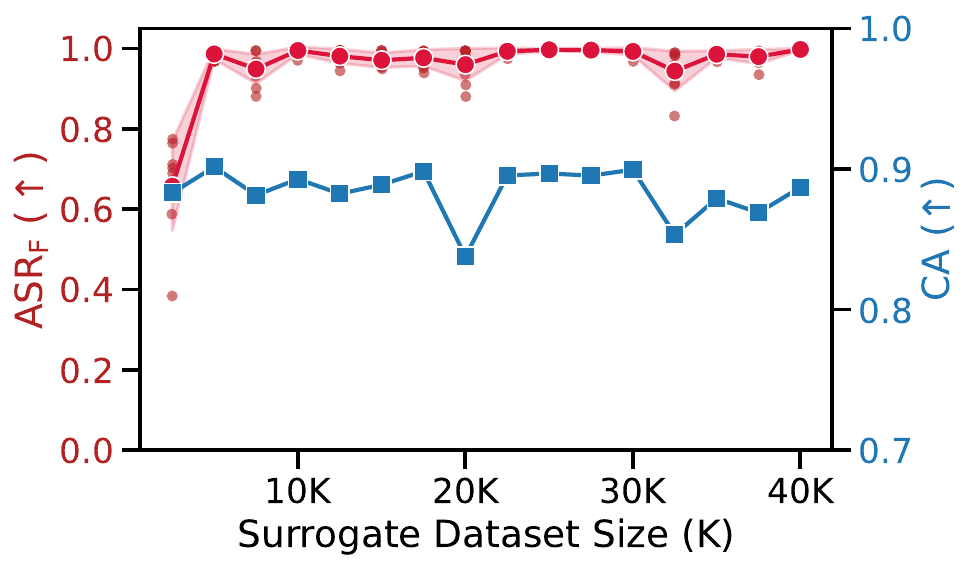}
        \caption{Surrogate Data Analysis}
        \label{fig:train_scale_ablation}
    \end{subfigure}
    \begin{subfigure}[b]{0.235\textwidth}
        \centering
        \includegraphics[width=\textwidth]{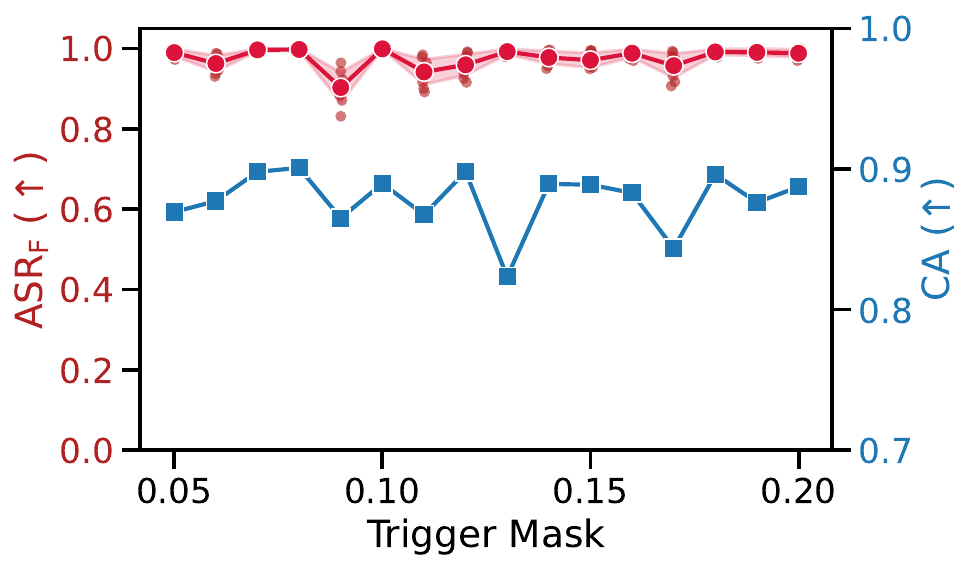}
        \caption{Trigger Mask Analysis}
        \label{fig:trigger_mask_ablation}
    \end{subfigure}
    \caption{Ablation study of \textbf{Lilith}.}
    \label{fig:ablation_study}
    \vspace{-1em}
\end{figure}

\begin{figure}
    \centering
    \includegraphics[width=1.0\linewidth]{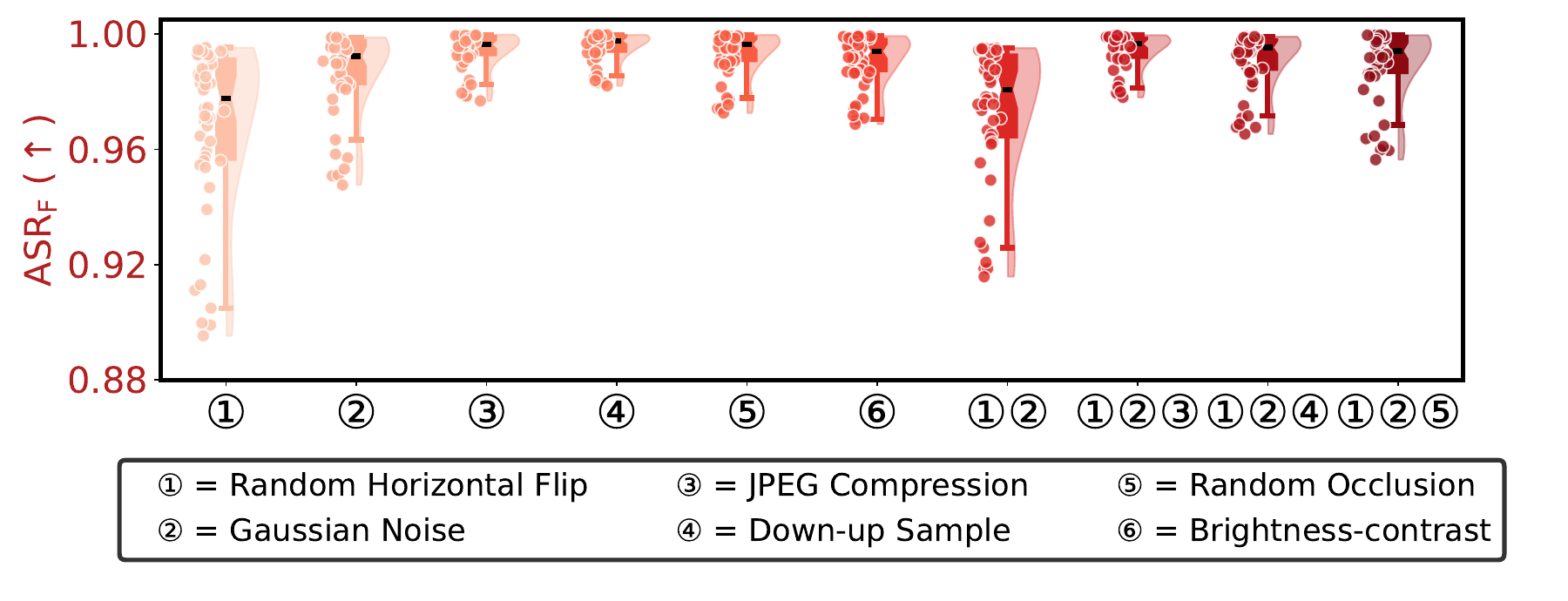}
    \caption{Robustness of \textbf{Lilith} under input augmentations.}
    \label{fig:augmentation}
    \vspace{-1em}
\end{figure}

\begin{figure}
    \centering
    \includegraphics[width=0.80\linewidth]{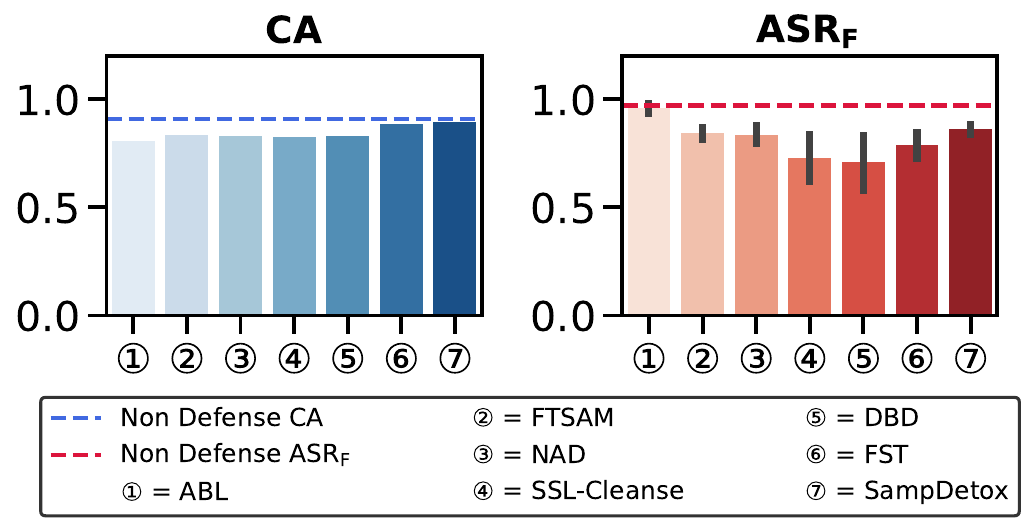}
    \caption{Robustness of \textbf{Lilith} under mitigation defenses.}
    \label{fig:mitigation}
    \vspace{-1em}
\end{figure}

\noindent\textbf{Sensitivity and Resource Cost.}
Figure~\ref{fig:ablation_study} examines surrogate-data size and trigger
strength. More surrogate data improves $\asr_F$ with little change in
clean accuracy, although \textbf{Lilith} remains effective with the smallest
evaluated set. Moderate blending best balances activation and perceptual
distortion, while stronger blending yields diminishing gains. Under the
default setting, anchor optimization takes about 15 minutes with
1.8\,GB peak GPU memory, and family proposal and trigger synthesis
complete within seconds.


\section{Conclusion}
We study \emph{backdoor generalization under training--inference trigger
shift}, exposing a blind spot in evaluations centered on exact trigger
reuse. We introduce \textbf{Lilith}, a black-box anchor-to-family
framework that implants a vulnerability with one training trigger and
extends it to an unseen inference family. Our analysis identifies
anchor clearance, family reach, and surrogate--victim transfer as the
key geometric conditions governing this generalization. Experiments
across datasets, architectures, and attack conditions show that
singleton poisoning supports reliable family-wise activation, with
representation alignment emerging as the decisive mechanism. These
findings shift the study of backdoors from memorizing a visible artifact
toward generalizing malicious behavior. They further suggest that
recovering or suppressing one trigger may be insufficient to
characterize the remaining post-deployment attack surface, motivating
trigger-shift evaluation and representation-level mitigation.
Future backdoor evaluation should therefore move beyond exact-trigger
protocols and account for training--inference trigger shift. More
broadly, studying malicious behavior under distribution shift may offer
a more realistic basis for evaluating attacks and defenses in practical
ML systems.




\appendix
\section{Proofs}
\label{app:proofs}

We first establish a basic property of the target margin. For model $m$, define
\begin{equation*}
    M_{m,t}(z)
    =
    s_{m,t}(z)-\max_{j\neq t}s_{m,j}(z).
\end{equation*}

\begin{proposition}[\textbf{Lipschitz Target Margin}]
\label{prop:margin-lipschitz}
If every score $s_{m,j}$ is $L_m$-Lipschitz on a set $\mathcal{N}$, then $M_{m,t}$ is $2L_m$-Lipschitz on $\mathcal{N}$.
\end{proposition}

\begin{proof}
Take any $z,z'\in\mathcal{N}$. We have
\begin{align*}
    |M_{m,t}(z)-M_{m,t}(z')|
    &\leq
    |s_{m,t}(z)-s_{m,t}(z')|
    \\
    &\quad+
    \left|
        \max_{j\neq t}s_{m,j}(z)
        -
        \max_{j\neq t}s_{m,j}(z')
    \right|.
\end{align*}
The first term is at most $L_m\|z-z'\|_2$. For the second term, the elementary inequality
\begin{equation*}
    \left|\max_j a_j-\max_j b_j\right|
    \leq
    \max_j|a_j-b_j|
\end{equation*}
implies
\begin{equation*}
    \left|
        \max_{j\neq t}s_{m,j}(z)
        -
        \max_{j\neq t}s_{m,j}(z')
    \right|
    \leq
    L_m\|z-z'\|_2.
\end{equation*}
Adding both bounds yields
\begin{equation*}
    |M_{m,t}(z)-M_{m,t}(z')|
    \leq
    2L_m\|z-z'\|_2.
\end{equation*}
\end{proof}

\subsection{Proof of Proposition~\ref{prop:anchor-activation}}

\begin{proof}
Let
\begin{equation*}
    \mathcal{E}_A
    =
    \left\{
        \|Z_v(x,\theta_A)-\mu_v(\theta_A)\|_2
        \leq
        r_v(\theta_A,\delta_A)
    \right\}.
\end{equation*}
By the definition of the concentration radius in Eq.~\eqref{eq:radius},
\begin{equation*}
    \Prb_x[\mathcal{E}_A]\geq 1-\delta_A.
\end{equation*}
On $\mathcal{E}_A$, Proposition~\ref{prop:margin-lipschitz} gives
\begin{align*}
    M_{v,t}(Z_v(x,\theta_A))
    &\geq
    M_{v,t}(\mu_v(\theta_A))
    \\
    &\quad-
    2L_v
    \|Z_v(x,\theta_A)-\mu_v(\theta_A)\|_2
    \\
    &\geq
    \gamma_{v,A}
    -
    2L_v r_v(\theta_A,\delta_A).
\end{align*}
The right-hand side is positive by Eq.~\eqref{eq:anchor-condition}. A positive target margin means that the target score exceeds every non-target score. Therefore,
\begin{equation*}
    f_v^\dagger(T_{\theta_A}(x))=t
\end{equation*}
for every $x$ in $\mathcal{E}_A$. It follows that
\begin{equation*}
    \asr_A
    \geq
    \Prb_x[\mathcal{E}_A]
    \geq
    1-\delta_A.
\end{equation*}
\end{proof}

\subsection{Proof of Proposition~\ref{prop:parameter-stability}}

\begin{proof}
Let $\theta=\theta_A+U\alpha$ with $\|\alpha\|_2\leq\eta$. Since $U$ has orthonormal columns, $\|\theta-\theta_A\|_2=\|U\alpha\|_2=\|\alpha\|_2\leq\eta$. Assumption~\ref{ass:trigger-manifold} and the mean-value theorem imply
\begin{equation*}
    \|g(\theta)-g(\theta_A)\|_2
    \leq
    M_g\|\theta-\theta_A\|_2
    \leq
    M_g\eta.
\end{equation*}
The Lipschitz property of the application operator gives
\begin{equation*}
\begin{aligned}
    \|T_\theta(x)-T_{\theta_A}(x)\|_2
    &\leq
    L_B\|g(\theta)-g(\theta_A)\|_2
    \\
    &\leq
    L_BM_g\eta.
\end{aligned}
\end{equation*}
Applying local Lipschitz continuity of $\phi_m$ yields
\begin{equation*}
\begin{aligned}
    \|Z_m(x,\theta)-Z_m(x,\theta_A)\|_2
    &\leq
    L_{\phi,m}
    \|T_\theta(x)-T_{\theta_A}(x)\|_2
    \\
    &\leq
    L_{\phi,m}L_BM_g\eta
    \\
    &=C_m\eta.
\end{aligned}
\end{equation*}
This proves Eq.~\eqref{eq:pointwise-stability}.

We next bound the center displacement. Jensen's inequality gives
\begin{align*}
    a_m(\theta;\theta_A)
    &=
    \left\|
        \E_x[Z_m(x,\theta)-Z_m(x,\theta_A)]
    \right\|_2
    \\
    &\leq
    \E_x
    \|Z_m(x,\theta)-Z_m(x,\theta_A)\|_2
    \\
    &\leq
    C_m\eta.
\end{align*}
This proves Eq.~\eqref{eq:center-bound}.

Let
\begin{equation*}
    \mathcal{E}_{m,A}
    =
    \left\{
        \|Z_m(x,\theta_A)-\mu_m(\theta_A)\|_2
        \leq
        r_m(\theta_A,\delta_A)
    \right\}.
\end{equation*}
The event has probability at least $1-\delta_A$. On this event, the triangle inequality gives
\begin{align*}
    \|Z_m(x,\theta)-\mu_m(\theta)\|_2
    &\leq
    \|Z_m(x,\theta)-Z_m(x,\theta_A)\|_2
    \\
    &\quad+
    \|Z_m(x,\theta_A)-\mu_m(\theta_A)\|_2
    \\
    &\quad+
    \|\mu_m(\theta_A)-\mu_m(\theta)\|_2
    \\
    &\leq
    C_m\eta
    +r_m(\theta_A,\delta_A)
    +C_m\eta.
\end{align*}
Hence a radius of $r_m(\theta_A,\delta_A)+2C_m\eta$ contains the variant representations with probability at least $1-\delta_A$. By the minimality in Eq.~\eqref{eq:radius},
\begin{equation*}
    r_m(\theta,\delta_A)
    \leq
    r_m(\theta_A,\delta_A)+2C_m\eta.
\end{equation*}
Combining this bound with Eq.~\eqref{eq:center-bound} yields
\begin{align*}
    \rho_m(\theta,\delta_A)
    &=
    a_m(\theta;\theta_A)+r_m(\theta,\delta_A)
    \\
    &\leq
    r_m(\theta_A,\delta_A)+3C_m\eta.
\end{align*}
This proves Eqs.~\eqref{eq:radius-bound} and~\eqref{eq:reach-bound}.
\end{proof}

\subsection{Proof of Theorem~\ref{thm:family-activation}}

\begin{proof}
Fix a family member $\theta$. Let
\begin{equation*}
    \mathcal{E}_\theta
    =
    \left\{
        \|Z_v(x,\theta)-\mu_v(\theta)\|_2
        \leq
        r_v(\theta,\delta(\theta))
    \right\}.
\end{equation*}
By Eq.~\eqref{eq:radius},
\begin{equation*}
    \Prb_x[\mathcal{E}_\theta]
    \geq
    1-\delta(\theta).
\end{equation*}
On this event,
\begin{align*}
    \|Z_v(x,\theta)-\mu_v(\theta_A)\|_2
    &\leq
    \|Z_v(x,\theta)-\mu_v(\theta)\|_2
    \\
    &\quad+
    \|\mu_v(\theta)-\mu_v(\theta_A)\|_2
    \\
    &\leq
    r_v(\theta,\delta(\theta))
    +a_v(\theta;\theta_A)
    \\
    &=
    \rho_v(\theta,\delta(\theta)).
\end{align*}
Proposition~\ref{prop:margin-lipschitz} then gives
\begin{align*}
    M_{v,t}(Z_v(x,\theta))
    &\geq
    M_{v,t}(\mu_v(\theta_A))
    \\
    &\quad-
    2L_v
    \|Z_v(x,\theta)-\mu_v(\theta_A)\|_2
    \\
    &\geq
    \gamma_{v,A}
    -
    2L_v\rho_v(\theta,\delta(\theta)).
\end{align*}
The definition in Eq.~\eqref{eq:risk-ratio} implies that the final expression is positive whenever
\begin{equation*}
    \zeta_v(\theta,\delta(\theta))<1.
\end{equation*}
For such a family member, the prediction equals $t$ on $\mathcal{E}_\theta$. Therefore,
\begin{equation*}
\begin{aligned}
    \Prb_x
    \bigl[f_v^\dagger(T_\theta(x))=t\bigr]
    \geq
    \bigl(1-\delta(\theta)\bigr)
    \ind\bigl\{
        \zeta_v(\theta,\delta(\theta))<1
    \bigr\}.
\end{aligned}
\end{equation*}
Taking expectation over $\theta\sim Q_{\mathrm{te}}$ proves Eq.~\eqref{eq:family-lower-bound}.
\end{proof}

\subsection{Proof of Corollary~\ref{cor:radius-family}}

\begin{proof}
Proposition~\ref{prop:parameter-stability} gives, for every $\theta\in\Theta_A(\eta)$,
\begin{equation*}
    \rho_v(\theta,\delta_A)
    \leq
    r_v(\theta_A,\delta_A)+3C_v\eta.
\end{equation*}
Equation~\eqref{eq:radius-family-condition} therefore implies
\begin{equation*}
    2L_v\rho_v(\theta,\delta_A)
    <
    \gamma_{v,A},
\end{equation*}
which is equivalent to $\zeta_v(\theta,\delta_A)<1$. Theorem~\ref{thm:family-activation}, applied with the constant failure probability $\delta(\theta)=\delta_A$, then yields
\begin{equation*}
    \asr_F(Q_{\mathrm{te}})
    \geq
    \E_{\theta\sim Q_{\mathrm{te}}}[1-\delta_A]
    =
    1-\delta_A.
\end{equation*}
\end{proof}

\subsection{Proof of Corollary~\ref{cor:uniform-family}}

\begin{proof}
The uniform assumption implies
\begin{equation*}
    \ind\bigl\{
        \zeta_v(\theta,\bar\delta)<1
    \bigr\}=1
\end{equation*}
for $Q_{\mathrm{te}}$-almost every $\theta$. Applying Theorem~\ref{thm:family-activation} with $\delta(\theta)=\bar\delta$ yields
\begin{equation*}
    \asr_F(Q_{\mathrm{te}})
    \geq
    \E_{\theta\sim Q_{\mathrm{te}}}[1-\bar\delta]
    =
    1-\bar\delta.
\end{equation*}
Since $\asr_A\leq 1$,
\begin{align*}
    \tgg(Q_{\mathrm{te}})
    &=
    [\asr_A-\asr_F(Q_{\mathrm{te}})]_+
    \\
    &\leq
    [1-(1-\bar\delta)]_+
    \\
    &=\bar\delta.
\end{align*}
\end{proof}

\subsection{Proof of Corollary~\ref{cor:surrogate-transfer}}

\begin{proof}
For $Q_{\mathrm{te}}$-almost every $\theta$, Assumption~\ref{ass:geometry-transfer} gives
\begin{equation*}
    \zeta_v(\theta,\bar\delta)
    \leq
    \zeta_s(\theta,\bar\delta)+\kappa_{\mathrm{tr}}.
\end{equation*}
Equation~\eqref{eq:surrogate-screening} implies
\begin{equation*}
    \zeta_s(\theta,\bar\delta)+\kappa_{\mathrm{tr}}<1
\end{equation*}
for $Q_{\mathrm{te}}$-almost every $\theta$. Hence
\begin{equation*}
    \esssup_{\theta\sim Q_{\mathrm{te}}}
    \zeta_v(\theta,\bar\delta)<1.
\end{equation*}
Corollary~\ref{cor:uniform-family} then gives
\begin{equation*}
    \asr_F(Q_{\mathrm{te}})\geq 1-\bar\delta.
\end{equation*}
\end{proof}

\end{document}